\documentclass{article}

 \usepackage[main, final]{neurips_2026}

\usepackage[utf8]{inputenc} 
\usepackage[T1]{fontenc}    
\usepackage{hyperref}       
\usepackage{url}            
\usepackage{booktabs}       
\usepackage{amsfonts}       
\usepackage{nicefrac}       
\usepackage{microtype}      
\usepackage{xcolor}         

\title{Learning-Augmented Approximation for Unrelated-Machines Makespan Scheduling}

\usepackage{amsmath,amssymb}
\usepackage{mathtools}
\usepackage{xspace}
\usepackage[capitalise]{cleveref}
\usepackage{booktabs}
\usepackage{subcaption}
\usepackage{wrapfig}
\usepackage{amsthm}
\theoremstyle{plain}
\newtheorem{theorem}{Theorem}[section]

\newtheorem{corollary}[theorem]{Corollary}
\theoremstyle{definition}

\theoremstyle{remark}
\newtheorem{remark}[theorem]{Remark}
\usepackage[ruled,vlined,linesnumbered]{algorithm2e}
\SetCommentSty{textbf}
\DontPrintSemicolon
\SetAlCapSkip{0.2em}
\SetVlineSkip{0.3ex}
\crefname{algocf}{Algorithm}{Algorithms}
\Crefname{algocf}{Algorithm}{Algorithms}

\makeatletter
  \DeclareRobustCommand\onedot{\futurelet\@let@token\@onedot}
  \def\@onedot{\ifx\@let@token.\else.\null\fi\xspace}
\makeatother

\newcommand\ie{i.e\onedot}

\newcommand\etal{et al\onedot}

\usepackage{xparse}
\newcommand{\dom}{\operatorname{dom}}
\NewDocumentCommand{\predphi}{s}{%
  \IfBooleanTF{#1}{\smash{\widehat\phi}}{\widehat\phi}\xspace%
}
\newcommand{\OPT}{\mathrm{OPT}}

\author{%
  Kaito Baba\thanks{This work was conducted during an internship at Sorbonne Université.} \\
The University of Tokyo,\\Tokyo, Japan\\
\texttt{baba-kaito662@g.ecc.u-tokyo.ac.jp}
\And
  Evripidis Bampis \\
Sorbonne Université, CNRS, LIP6,\\F-75005 Paris, France\\
  \texttt{evripidis.bampis@lip6.fr}
  \And
  Giorgos Mitropoulos\\
  Sorbonne Université, CNRS, LIP6,\\F-75005 Paris, France\\
  \texttt{georgios.mitropoulos@lip6.fr}
}

\begin{document}

\maketitle
\setcounter{footnote}{0}

\begin{abstract}
 Recently, Antoniadis \etal (ICLR 2025) proposed a framework for incorporating predictions to approximate NP-hard selection problems. Despite its simplicity, this approach tightly matches theoretical lower bounds, making its generalization highly compelling. We address an open  question raised  in the work of Antoniadis \etal, concerning the extension of this approach to other important problems outside the class of selection problems, such as scheduling. We  develop a learning-augmented algorithm for the makespan minimization problem on unrelated machines, denoted by $R\|C_{\max}$. By using predictions of heavy job assignments, we achieve a polynomial-time $(1+\varepsilon)$-approximation for accurate predictions that smoothly degrades to a worst-case 2-approximation as the error increases. We conclude our work with an empirical analysis of our method.\footnote{Our code is available at \url{https://github.com/kAIto47802/rcmax-with-predictions}.}
\end{abstract}

\section{Introduction}

Scheduling is one of the central problems in combinatorial optimization~\citep{Pinedo2016,Brucker2007}.
A particularly fundamental model is scheduling on unrelated parallel machines, denoted by $R\|C_{\max}$~\citep{GRAHAM1979287}.
Here, we are given a set of $n$ jobs and a set of $m$ machines, and each job $j$ has a machine-dependent processing time $p_{ij}$ on machine $i$.
The goal is to assign each job to one of the machines so as to minimize the maximum load over all machines.
This problem models settings in which processing times vary across machines.
At the same time, it is a notoriously difficult NP-hard problem: the classical result of \citet{Lenstra1990} gives a polynomial-time $2$-approximation, and they showed that improving the approximation factor below $3/2$ is NP-hard.

In recent years, there has been growing interest in algorithms with predictions, or \emph{learning-augmented algorithms}~\citep{NEURIPS2018_73a427ba,pmlr-v80-lykouris18a}.
Motivated by the increasing availability of high-quality machine-learning predictions, this line of work studies how algorithms can exploit such predictions while retaining rigorous guarantees when they are inaccurate.
Most early work focused on online algorithms, where the standard objectives are \emph{consistency}, \emph{robustness}, and \emph{smoothness}: the algorithm should perform optimally when predictions are accurate, remain competitive when they are poor, and degrade smoothly as prediction error increases~\citep{10.1145/3528087,NEURIPS2023_82f0dae8}.

More recently, learning-augmented ideas have also begun to appear in approximation algorithms for NP-hard problems~\citep{NEURIPS2024_2db08b94,pmlr-v235-bampis24a,antoniadis2025approximation}.
In this offline setting, predictions may describe part of an optimal solution or some structural object that a classical approximation algorithm would otherwise need to guess or search over.
The main goal is to obtain near-optimal performance when the prediction is accurate, together with a guarantee that degrades smoothly with the prediction error.
Robustness can often be obtained by running a classical approximation algorithm in parallel, so the central challenge is to design algorithms that yield meaningful consistency and smoothness.

A recent work of~\citet{antoniadis2025approximation} develops such algorithms for a broad class of selection problems with linear objectives, including Vertex Cover and Steiner Tree, and provides approximation guarantees that vary smoothly with false-positive and false-negative prediction errors.
However, their framework does not cover problems beyond this linear-objective selection setting, and they explicitly identify constructing such prediction-dependent smooth guarantees for other optimization problems, such as scheduling, as an open direction.

In this work, we develop a learning-augmented approximation algorithm for $R\|C_{\max}$, addressing this open question posed by \citet{antoniadis2025approximation} in the context of scheduling.
While motivated by their open question, our work is neither an instantiation of their black-box framework nor a direct extension of their algorithms; our algorithm, prediction-error measures, and analysis are tailored to the assignment structure and bottleneck objective of $R\|C_{\max}$.
Our starting point is the classical heavy/short decomposition for unrelated-machine scheduling~\citep{Lenstra1990,mordechai2015optimization}.
For a target makespan $T$, the combinatorial difficulty lies in assignments with processing time greater than $\varepsilon T$; once these assignments are fixed, the remaining short jobs can be handled by a linear program (LP) and an appropriate rounding procedure, such as LST rounding \citep{Lenstra1990} and Shmoys--Tardos rounding~\citep{Shmoys1993}.
Classical schemes search over the heavy assignments to obtain a $(1+\varepsilon)\OPT$ guarantee, but this requires non-polynomial running time unless $m$ is fixed to be small.
We use the prediction to localize this search.

Our learning-augmented algorithm achieves a $(1+\varepsilon)\OPT$ guarantee with accurate predictions while running in time comparable to the classical polynomial-time $2$-approximation algorithm, thereby improving the makespan guarantee when $m$ is large. The guarantee degrades smoothly with the prediction error and, even with poor predictions, never exceeds the classical $2\OPT$ baseline. At a high level, false-positive predictions contribute to the makespan bound, while false-negative predictions can either be recovered by local search or left in the bound, with the constant search budget $K$ smoothly controlling this tradeoff. \cref{tbl:comparison-classical-ours} summarizes the comparison with the classical guarantees. Experiments show that, consistently with our theory, accurate predictions allow our algorithm to outperform the classical $2$-approximation baseline in makespan,
while running significantly faster than the classical $(1+\varepsilon)$-approximation based on heavy-assignment search.

\begin{table}[t]
  \centering
  \setlength{\tabcolsep}{2pt}
  \caption{
    Comparison with classical bounds for $\smash{R\|C_{\max}}$.
    Here, $\smash{R_{\OPT}^-}$, $\smash{R_{\OPT,K}^-}$, and $\smash{R_{\OPT}^+}$ are prediction error measures defined in \cref{sec:prediction-error-measures}: the first two capture false negatives, while $\smash{R_{\OPT}^+}$ captures false-positive overload.
    $\smash{R_{\OPT,K}^-}$ monotonically decreases with the search budget $K$, which is chosen as a constant parameter. $P\coloneqq\max_{i,j} p_{ij}$ denotes the largest processing time of the input.
  }
  \vspace*{3mm}
  \label{tbl:comparison-classical-ours}
  \scalebox{0.66}{
    \begin{tabular}{lccc}
      \toprule
      \textbf{Algorithm}
      & \textbf{Makespan guarantee}
      & \textbf{Running time}
      & \textbf{Assumption / Parameter} \\
      \midrule

      \makebox[17mm][l]{LST \citep{Lenstra1990}}
      & $2\OPT$
      & $O(\mathrm{poly}(m,n,\log P))$
      & --
      \\

      \makebox[17mm][l]{LST \citep{Lenstra1990}}
      & $(1+\varepsilon)\OPT$
      & $O((n+1)^{m/\varepsilon} \mathrm{poly}(m,n,\log P))$
      & fixed $m$
      \\

      \citet{mordechai2015optimization}
      & $(1+\varepsilon)T$
      & $O(2^k \mathrm{poly}(m,n,\log P,1/\varepsilon))$
      & \hspace{-2mm}$k\coloneqq\left|\{(i,j)\mid p_{ij}>\varepsilon T\}\right|$, given $T$
      \\

      \midrule

      Ours (\cref{thm:outer-layer-guarantee})
      & $
        \min\left\{
        \OPT
        +
        \max\{\varepsilon \OPT, R_{\OPT,K}^-\}
        +
        R_{\OPT}^+,
        \,
        2\OPT
        \right\}
      $
      & $
        O\!\left(
          (mn)^K
          \mathrm{poly}(m,n,\log P,1/\varepsilon)
        \right)
      $
      & --
      \\

      Ours (\cref{cor:outer-layer-guarantee-asymmetric})
      & $
        \min\left\{
          (1+\varepsilon)\OPT
          +
          R_{\OPT}^+,
          \,
          2\OPT
        \right\}
      $
      & $
        O\!\left(
        (mn)^K
        \mathrm{poly}(m,n,\log P,1/\varepsilon)
        \right)
      $
      & $K\ge |\mathrm{Miss}_{\OPT}(\widehat{\phi},\sigma^\star)|$
      \\

      Ours (\cref{cor:outer-layer-guarantee-k-free})
      & $
        \min\left\{
        \OPT
        +
        \max\{\varepsilon \OPT, R_{\OPT}^-\}
        +
        R_{\OPT}^+,
        \,
        2\OPT
        \right\}
      $
      & $
        O\!\left(
        \mathrm{poly}(m,n,\log P,1/\varepsilon)
        \right)
      $
      & --
      \\
      \bottomrule
    \end{tabular}
  }
\end{table}

Our contributions can be summarized as follows:
\vspace{-2mm}
\begin{itemize}
    \item We introduce a learning-augmented framework for $R\|C_{\max}$, addressing the open question of \citet{antoniadis2025approximation} in the context of scheduling.
    Rather than searching globally over heavy assignments, our approach localizes the search around the prediction.
    \item Our guarantee exhibits the desired learning-augmented behavior.
    Accurate predictions yield a $(1+\varepsilon)\OPT$ makespan bound in polynomial time, and it degrades smoothly with the prediction error; even with poor predictions, it never exceeds the classical $2\OPT$ baseline.
    \item We also provide the smooth tradeoff between running time and makespan quality.
    Larger search budgets $K$ allow the algorithm to repair more missed heavy assignments, while $K=0$ yields a fast, purely prediction-based variant without local search.
    \item Numerical experiments show that, consistent with our theory, accurate predictions improve over the classical $2$-approximation baseline, while our algorithm runs significantly faster than classical $(1+\varepsilon)$-type heavy-assignment search.
\end{itemize}

\vspace{-1mm}
We also extend our prediction-localized search approach to total weighted completion time on unrelated machines.
This extension yields a guarantee that degrades smoothly with prediction error while preserving the approximation ratio of a classical baseline; see \cref{apx:extension_to_completion_time}.

\section{Related work}

Here we review the prior work most relevant to ours; further discussion is provided in \cref{apx:extended-related-work}.

\vspace{-8pt}
\paragraph{Learning-augmented online algorithms.}
Learning-augmented algorithms were first studied extensively in the online setting, where predictions are used to improve performance beyond worst-case guarantees while retaining robustness when the predictions are inaccurate~\citep{10.1145/3528087}.
Starting with the works of \citet{pmlr-v80-lykouris18a,NEURIPS2018_73a427ba} on caching, ski rental, and online scheduling, subsequent work developed sharper guarantees for caching~\citep{doi:10.1137/1.9781611975994.112,wei:LIPIcs.APPROX/RANDOM.2020.60}, primal-dual frameworks for online covering problems~\citep{NEURIPS2020_e834cb11}, and more general online covering linear and semidefinite programs~\citep{NEURIPS2022_fc5a1845}.

\vspace{-7pt}
\paragraph{Learning-augmented approximation algorithms.}
More recently, learning-augmented ideas have been applied to offline approximation algorithms for NP-hard problems.
Recent results include learning-augmented approximation algorithms for Max-Cut and related CSPs~\citep{NEURIPS2024_2db08b94}, dense NP-hard problems~\citep{pmlr-v235-bampis24a}, Maximum Independent Set~\citep{braverman_et_al:LIPIcs.APPROX/RANDOM.2024.24}, and permutation problems~\citep{pmlr-v267-bampis25a}.
Recently, \citet{antoniadis2025approximation} gave black-box algorithms for selection problems with linear objectives, including Vertex Cover, Steiner Tree, Min-Weight Perfect Matching, Knapsack, and Clique, with approximation guarantees that degrade smoothly with false-negative and false-positive prediction errors.
Their framework, however, does not cover scheduling problems such as $R\|C_{\max}$, which are not selection problems and whose objective is makespan minimization rather than the total weight of a selected set.
They leave scheduling as an open question for such smooth guarantees, which we address for $R\|C_{\max}$.

\vspace{-7pt}
\paragraph{Scheduling and heavy-assignment search.}
In the context of classical approximation algorithms without prediction, \citet{Lenstra1990} gave a polynomial-time $2$-approximation for $R\|C_{\max}$, as well as a $(1 + \varepsilon)\OPT$ approximation scheme when the number of machines $m$ is fixed.
The latter relies on exhaustive search over heavy assignments and LP rounding for the remaining short jobs.
For a given threshold $T$, \citet{mordechai2015optimization} later gave a parameterized
approximation scheme whose parameter is
$k\coloneqq\bigl|\{(i,j)\mid p_{ij}>\varepsilon T\}\bigr|$.
However, these $(1+\varepsilon)$-type guarantees are not polynomial-time unless $m$ is fixed to be small.
Our work uses predictions to go beyond these fixed-$m$ regimes: when the predictions are accurate, we obtain a $(1+\varepsilon)\OPT$ guarantee with running time comparable to that of the classical polynomial-time $2$-approximation algorithm.

\section{Problem setup}

\paragraph{Problem definition.}
We consider the \emph{makespan minimization} problem on \emph{unrelated-machines}.
There are $m$ machines and $n$ jobs. 
For every machine $i\in [m]$ and job $j\in [n]$\footnote{Here and throughout, for a positive integer $r$, we write $[r]\coloneqq\{1,2,\dots,r\}$.}, let $p_{ij}\in \mathbb{R}_{>0}$ denote the processing time of job $j$ on machine $i$.
A schedule is a function $\sigma\colon[n]\to [m]$, and its makespan is
$
C_{\max}(\sigma) \coloneqq \max_{i\in [m]} \sum_{j:\sigma(j)=i} p_{ij}
$.
The goal is to find a schedule $\sigma$ minimizing $C_{\max}(\sigma)$.

\paragraph{Prediction model.}
The algorithm is given as input a predicted partial assignment
$
\predphi*\colon[n]\to [m]\cup\{\bot\},
$
where $\predphi(j)=i$ means that the prediction suggests assigning job $j$ to machine $i$, and $\predphi(j)=\bot$ means that no assignment is predicted for job $j$.
Intuitively, $\predphi$ is intended to capture part of the heavy-assignment structure of a good schedule.
We let
$
 \dom(\predphi) \coloneqq \{j\in [n] \mid \predphi(j)\neq \bot\}
$
denote the set of jobs for which the prediction specifies a machine.
See \cref{rem:full-assignment-prediction} for further discussion.
We discuss the motivation for this prediction model and compare it with alternatives in \cref{apx:prediction_choice}.

\section{Algorithm}
\label{sec:algorithm}

Our algorithm has the classical two-layer structure used in approximation schemes for $R\|C_{\max}$ \citep{Lenstra1990,mordechai2015optimization}: an \emph{inner} layer that solves a completion problem for a fixed threshold makespan $T$, and an \emph{outer} layer that ranges over the relevant threshold values.
However, in our algorithm, both layers are prediction-aware, which we explain in the following sections.

\begin{algorithm}[t]
  \caption{Inner layer: \textsc{Prediction-Guided-Completion}$(T,\varepsilon,K,\predphi)$}
  \label{alg:inner-layer}
  \KwIn{target makespan $T$, parameter $\varepsilon\in(0,1)$, search budget $K$, and predicted assignment $\predphi$}

  \vspace{4pt}
  \tcp{Step 1: Keep only the predicted assignments that are $T$-heavy.}
  \vspace{1pt}Compute $\predphi*_T$ according to \cref{eq:heavy-projection}\;
  Initialize the candidate set $\mathcal{C}\gets \emptyset$\;

  \vspace{4pt}
  \tcp{Step 2: Search only for missing heavy assignments.}
  \vspace{1pt}\ForEach{augmentation $A$ of $\predphi*_T$ with budget $K$ satisfying \cref{eq:augmentation-constraint}}{
    $\psi_A \gets \widehat{\phi}_T \cup A$\;

    \vspace{4pt}
    \tcp{Step 3: Fix the heavy assignments and solve the residual LP over relevant thresholds.}
    \ForEach{threshold $\tau \in \Lambda(T)$ defined in \cref{eq:residual-lp-thresholds}}{
      Solve the residual LP \cref{eq:residual-lp} with threshold $\tau$ for $\psi_A$\;

      \If{the LP is feasible with optimum value $L_{\psi_A}$}{
        \vspace{4pt}
        \tcp{Step 4: Round the solution of the residual LP.}
        $S_A \gets$ the schedule obtained by Shmoys--Tardos rounding~\citep{Shmoys1993}\;
        $\mathcal{C}\gets \mathcal{C}\cup\{S_A\}$\;
      }
    }
  }

  \vspace{4pt}
  \tcp{Step 5: Return the best schedule among all candidates.}
  \Return{the schedule in $\mathcal{C}$ that minimizes makespan}\;
\end{algorithm}

\subsection{Inner layer: prediction-localized heavy-assignment completion}

We now describe the inner layer of our algorithm, which is the main algorithmic component where prediction enters the search procedure. The pseudocode is given in \cref{alg:inner-layer}.

\paragraph{Step 1: Keep only the predicted assignments that are $T$-heavy.}
The prediction $\predphi*$ may contain assignments that are irrelevant at scale $T$.
Accordingly, the algorithm forms the $T$-heavy projection:
\begin{equation}\label{eq:heavy-projection}
    \predphi*_T(j)=
    \begin{cases}
    \predphi*(j), & \text{if } \predphi*(j)\neq \bot \text{ and } p_{\predphi(j),j}>\varepsilon T,\\
    \bot, & \text{otherwise.}
    \end{cases}
\end{equation}
The intuition is that only jobs with processing time exceeding $\varepsilon T$ need special treatment.
Assignments with processing time at most $\varepsilon T$ are left to the LP, since the rounding incurs only an additive $\varepsilon T$ loss.

\begin{remark}\label[remark]{rem:full-assignment-prediction}
    A full-assignment prediction $\predphi\colon [n]\to [m]$ is a special case of our framework:
    once $T$ is fixed, one simply keeps the predicted assignments with $\smash{p_{\predphi(j),j}>\varepsilon T}$.
    Allowing abstention is more flexible, since uncertain heavy jobs can be left to the local search rather than fixed incorrectly.
\end{remark}

\begin{remark}
  In the present formulation, every predicted assignment with processing time larger than $\varepsilon T$ is trusted and fixed.
  This makes the distinction between false-positive and false-negative errors especially transparent.
  A variant that allows editing or deleting some predictions will be considered in \cref{apx:editable-variant}.
  Conceptually, such edits simply trade one type of prediction error for the other.
\end{remark}

\paragraph{Step 2: Search only for missing heavy assignments.}
The projected prediction $\predphi*_T$ may miss some heavy assignments in a good schedule.
The algorithm therefore enumerates augmentations $A\colon[n]\to [m]\cup\{\bot\}$ that add at most $K$ extra heavy assignments, where
\begin{equation}\label{eq:augmentation-constraint}
  \operatorname{dom}(A)\cap \operatorname{dom}(\predphi_T)=\emptyset,
  \;\text{and}\;
  |\operatorname{dom}(A)|\le K,
  \;\text{and}\;
  \ p_{A(j),j}>\varepsilon T \text{ for every } j\in \operatorname{dom}(A).
\end{equation}
Thus, the algorithm replaces the classical global heavy-assignment search~\citep{Lenstra1990,mordechai2015optimization} by a search for at most $K$ missing heavy assignments.
Even in the worst case, the necessary budget $K$ never exceeds that of the classical heavy-assignment search.
We also note that our algorithm still works if $K$ is insufficient; the remaining missed heavy assignments are smoothly reflected in the makespan guarantees in \cref{thm:inner-layer-guarantee,thm:outer-layer-guarantee}.
Thus, $K$ is a freely chosen constant parameter, not a problem- or prediction-dependent quantity.
This role is analogous to the trust parameter $\lambda$ used in prior learning-augmented algorithms~\citep{NEURIPS2018_73a427ba,NEURIPS2020_e834cb11}: it controls how much the algorithm corrects or trusts the prediction.

\paragraph{Step 3: Fix the heavy assignments and solve the residual LP over relevant thresholds.}
Once $\psi_A \coloneqq \predphi*_T \cup A$ is fixed, the algorithm computes the fixed machine loads from $\operatorname{dom}(\psi_A)$ and solves an LP on the remaining jobs.
However, if the prediction has errors and $K$ does not recover all missed heavy assignments, some residual jobs may still have processing time exceeding $\varepsilon T$.
Hence, unlike in the classical setting, the residual LP cannot be restricted to edges $(i, j)$ with $p_{ij}\le \varepsilon T$ alone.

For this reason, the algorithm considers a threshold parameter $\tau\ge \varepsilon T$ and lets the residual LP use all edges $(i,j)$ with $p_{ij}\le \tau$.
Since $\tau$ affects the LP only through the admissible edge set, it changes only when $\tau$ crosses some processing time $p_{ij}$.
Thus, it suffices to consider the discrete candidate set
\begin{equation}\label{eq:residual-lp-thresholds}
  \Lambda(T)\coloneqq \{\varepsilon T\}\cup\{p_{ij}\mid p_{ij}\ge \varepsilon T\}.
\end{equation}
For each $\tau\!\in\!\Lambda(T)$, the algorithm solves an LP over edges with $p_{ij}\!\le\! \tau$ to minimize the load bound $L$:
\begin{equation}\label{eq:residual-lp}
  \begin{aligned}
    \text{minimize} \qquad & L\\[-3pt]
    \text{subject to} \qquad
    & \sum_{i=1}^m x_{ij} = 1
    && \forall j\notin \operatorname{dom}(\psi_A),\\[-1pt]
    & \hspace{-10.6pt}\sum_{j\notin \operatorname{dom}(\psi_A)} p_{ij}x_{ij} \le L-\ell_i^{\mathrm{fix}}(\psi_A)
    && \forall i\in [m],\\[-2pt]
    & x_{ij}=0
    && \forall i,j \text{ with } j\notin \operatorname{dom}(\psi_A) \text{ and } p_{ij}>\tau,\\[-1pt]
    & x_{ij}\ge 0
    && \forall i,j \text{ with } j\notin \operatorname{dom}(\psi_A),
  \end{aligned}
\end{equation}
where
$
  \ell_i^{\mathrm{fix}}(\psi_A)\!\coloneqq\!\sum_{j:\psi_A(j)=i} p_{ij}
$
is the fixed load on machine $i$ induced by $\psi_A$.
Note that the threshold $\tau$ only affects the residual LP; heavy assignments $\predphi*_T(j)$ are still defined by $\varepsilon T$.

\paragraph{Step 4: Round the solution of the residual LP.}
If the residual LP is feasible, the algorithm applies the Shmoys--Tardos rounding~\citep{Shmoys1993} to obtain an integral schedule.\footnote{This rounding generalizes the LST rounding~\citep{Lenstra1990}. In our setting, both yield the same guarantee.}
Because every residual edge has processing time at most $\tau$, the rounding increases the load of each machine by at most $\tau$.
Therefore, the final makespan of the schedule extending $\psi_A$ is at most
$
L_{\psi_A} + \tau
$,
where $L_{\psi_A}$ is the optimum value of the residual LP, if the LP is feasible.
This is exactly why the heavy/short decomposition is useful and why the algorithm does not need predictions for short assignments.

\paragraph{Step 5: Return the best schedule among all candidates.}
The algorithm compares all schedules produced by feasible augmentations, and returns the one with minimum makespan.

\begin{algorithm}[t]
  \caption{Outer layer: \textsc{Prediction-Localized-Search}$(\varepsilon,K,\widehat{\phi})$}
  \label{alg:outer-layer}
  \KwIn{parameter $\varepsilon\in(0,1)$, search budget $K$, and predicted assignment $\widehat{\phi}$}

  \vspace{4pt}
  \tcp{Step 1: Compute the baseline upper bound.}
  Compute a baseline schedule $S_0$ using the classical $2$-approximation~\citep{Lenstra1990}\;
  Set $B\gets C_{\max}(S_0)$, and initialize the candidate set $\mathcal{C}\gets \{S_0\}$\;

  \vspace{4pt}
  \tcp{Step 2: Run the inner procedure on all critical thresholds.}
  \ForEach{$T\in \mathcal T$, where $\mathcal{T}$ is defined in \cref{eq:critical-thresholds}}{
      $S_T \gets \textsc{Prediction-Guided-Completion}(T,\varepsilon,K,\predphi*)$\;
      $\mathcal{C} \gets \mathcal{C} \cup \{S_T\}$\;
  }

  \vspace{4pt}
  \tcp{Step 3: Return the best schedule among all candidates.}
  \Return{the schedule in $\mathcal{C}$ that minimizes makespan}\;
\end{algorithm}

\subsection{The outer layer: synchronizing the threshold search with the prediction}

The inner layer assumes that the target makespan $T$ is fixed.
The outer layer removes this assumption.
Unlike the classical inner routine, our prediction-guided inner layer does not certify feasibility at a fixed threshold $T$, since failure may reflect an inaccurate prediction rather than true infeasibility; consequently, the classical binary search over $T$ is no longer applicable.
Nevertheless, we can still design the outer layer without increasing the overall running-time order, by enumerating the polynomially many critical threshold values at which the prediction-induced heavy-assignment structure changes.
The pseudocode is given in \cref{alg:outer-layer}.

\paragraph{Step 1: Compute the baseline upper bound.}
The algorithm first computes the classical $2$-approximation schedule $S_0$~\citep{Lenstra1990} and sets $B \coloneqq C_{\max}(S_0)$.
Since $B\le 2\OPT$, the unknown optimum lies in the interval $[B/2,B]$.

\paragraph{Step 2: Run the inner procedure on all critical thresholds.}
The heavy/short classification changes only when $\varepsilon T$ crosses some processing time $p_{ij}$, \ie, when $T$ crosses some critical threshold $T=p_{ij}/\varepsilon$.
Therefore, the outer algorithm constructs the critical threshold set
\begin{equation}\label{eq:critical-thresholds}
  \mathcal T
  =
  \left\{
    \frac{B}{2}
  \right\}
  \cup
  \left\{
  \frac{p_{ij}}{\varepsilon} \;\middle|\;
  \frac{B}{2} \le \frac{p_{ij}}{\varepsilon}\le B
  \right\}
  \cup
  \{B\},
\end{equation}
and runs the inner procedure for every $T\in \mathcal{T}$:
$
\textsc{Prediction-Guided-Completion}(T,\varepsilon,K,\predphi*)
$.

\paragraph{Step 3: Return the best schedule among all candidates.}
Finally, the algorithm returns the best schedule among:
(i) the baseline schedule $S_0$, and
(ii) all schedules produced at all critical thresholds.

\section{Analysis}
\label{sec:analysis}

\subsection{Prediction error measures}
\label{sec:prediction-error-measures}

Following \citet{antoniadis2025approximation}, we decompose prediction error into false-negative and false-positive components.
Because their framework is designed for selection problems, we define the corresponding notions for \(R\|C_{\max}\).
For any schedule $\sigma\colon [n]\to[m]$,
define the \emph{false-negative threshold} and \emph{false-positive overload}:
\begin{equation}\label{eq:fn-fp-error-measures}
  R_T^-(\widehat{\phi},\sigma)
  \coloneqq
  \!\!\!\max_{j\in \mathrm{Miss}_T(\widehat{\phi},\sigma)}
  \!\!\! p_{\sigma(j),j},
  \qquad
  R_T^+(\widehat{\phi},\sigma)
  \coloneqq
  \max_{i\in [m]}
  \!\!\!\sum_{\substack{j:\ \widehat{\phi}_T(j)=i\\ \widehat{\phi}_T(j)\neq \sigma(j)}}\!\!\! p_{ij},
\end{equation}
where $\!\smash{\mathrm{Miss}_T(\widehat{\phi},\sigma)}\!$ is the \emph{missed heavy jobs} defined as
$
  \smash{\mathrm{Miss}_T(\widehat{\phi},\sigma)}
  \!\!\coloneqq\!\!
  \{j\!\in\! [n]\!\mid\! p_{\sigma(j),j}\!>\!\varepsilon T \text{ and}\allowbreak \predphi*_T(j)\!=\!\bot\}
$.
Here, we treat $R_T^-(\widehat{\phi},\sigma)=0$ if
$\mathrm{Miss}_T(\widehat{\phi},\sigma)=\emptyset$.
Intuitively, $R_T^-(\widehat{\phi},\sigma)$ measures the largest heavy assignment missed by the prediction, while $\smash{R_T^+(\widehat{\phi},\sigma)}$ measures the maximum overload caused by
predicted heavy assignments that disagree with \(\sigma\).
We further define the \emph{$K$-residual false-negative threshold}:
\begin{equation}
  R_{T,K}^-(\predphi,\sigma)
  \coloneqq
  \!\!\!\min_{\substack{U\subseteq \mathrm{Miss}_T(\widehat{\phi},\sigma)\\ |U|\le K}}
  \;
  \max_{j\in \mathrm{Miss}_T(\widehat{\phi},\sigma)\setminus U}
  \!\!\!p_{\sigma(j),j},
\end{equation}
with the convention that the maximum over the empty set is $0$.
Equivalently, $\smash{R_{T,K}^-(\widehat{\phi},\sigma)}$ is the $(K+1)$-st largest processing time among the missed heavy jobs.
This represents the largest processing time among the missed heavy jobs that remain after optimally recovering up to $K$ of them.
In particular, when $K = 0$, it is the same as the original false-negative threshold:
$
R_{T,0}^-(\widehat{\phi},\sigma)=R_T^-(\widehat{\phi},\sigma)
$.

\subsection{Guarantees for the inner layer}

\begin{theorem}[Inner-layer guarantee]
  \label{thm:inner-layer-guarantee}
  Fix $\varepsilon\in(0,1)$, a target value $T>0$, a search budget $K\ge 0$, and a prediction $\widehat{\phi}$.
  Then for any schedule $\sigma$, \cref{alg:inner-layer} returns a feasible schedule $S$ satisfying
  \begin{equation}
    C_{\max}(S)
    \le
    C_{\max}(\sigma)
    + \max\{\varepsilon T, R_{T, K}^-(\widehat{\phi},\sigma)\}
    + R_T^+(\widehat{\phi},\sigma).
  \end{equation}
  The running time of the algorithm is
  $
    O((mn)^K\mathrm{poly}(m,n,\log P,1/\varepsilon))
  $,
  where $P\coloneqq\max_{i,j} p_{ij}$ is the largest processing time.
\end{theorem}

Here we present a proof sketch of \cref{thm:inner-layer-guarantee}. The full proof is deferred to \cref{apx:full-proof-of-inner-layer}.

\begin{proof}[Proof sketch]
  By the definition of $\smash{R_{T,K}^-(\predphi*,\sigma)}$, there exists a subset
  $U^\star\subseteq \mathrm{Miss}_{T}(\predphi*,\sigma)$ with $|U^\star|\le K$ such that $
    \max_{j\in \mathrm{Miss}_{T}(\widehat{\phi},\sigma)\setminus U^\star} p_{\sigma(j),j}
    =
    R_{T,K}^-(\predphi*,\sigma)
  $.
  Define an augmentation $A^\star(j)=\sigma(j)$ for $j\in U^\star$ and $A^\star(j)=\bot$
  otherwise. Since $U^\star$ consists only of missed heavy assignments,
  we can confirm that $A^\star$ satisfies \cref{eq:augmentation-constraint} and is one of the augmentations considered by the algorithm.

  Let
  $\smash{\psi^\star\coloneqq \predphi*_T\cup A^\star}$ and
  $
  \smash{\tau^\star\coloneqq
  \max\{\varepsilon T,R_{T,K}^-(\predphi*,\sigma)\}}
  $.
  By construction, we can confirm $\tau^\star\in\Lambda(T)$. Consider the residual LP for
  $(\psi^\star,\tau^\star)$, and assign every residual job $j$ to
  $\sigma(j)$.
  Then we can confirm that every edge $(\sigma(j),j)$ is allowed in the residual LP with $\tau^\star$:
  if
  $p_{\sigma(j),j}\le \varepsilon T$ this is immediate, while otherwise $j$ is
  a missed heavy job not recovered by $U^\star$, and hence
  $\smash{p_{\sigma(j),j}\le R_{T,K}^-(\predphi*,\sigma)\le \tau^\star}$.

  Consider the assignments that follows $\smash{\psi^\star}$ on fixed jobs and $\sigma$ on residual jobs.
  We can confirm that this assignment is bounded by
  $
    \smash{C_{\max}(\sigma)+R_T^+(\widehat{\phi},\sigma)}
  $.
  Thus, the residual LP for $(\psi^\star,\tau^\star)$ is feasible with load bound at most
  $\smash{C_{\max}(\sigma)+R_T^+(\widehat{\phi},\sigma)}$.

  Let $\smash{L_{\psi^\star}^{\tau^\star}}$ be the returned optimum of the residual LP for $(\psi^\star,\tau^\star)$.
  By the classical rounding argument (Theorem 2.1 in \citet{Shmoys1993}), this  yields an integral
  schedule of makespan at most $\smash{L_{\psi^\star}^{\tau^\star} + \tau^\star}$.
  Hence, for $(A^\star,\tau^\star)$, the algorithm obtains a schedule of makespan at most
  \begin{equation}
    \smash{C_{\max}(\sigma)
    +R_T^+(\predphi*,\sigma)
    +\max\{\varepsilon T,R_{T,K}^-(\predphi*,\sigma)\}.}
  \end{equation}
  Since the algorithm returns the best candidate, the same bound holds for its
  output.

  The running time follows from the enumeration bounds. The number of
  augmentations is at most
  $
    \sum_{s=0}^{K}\binom{n}{s}m^s
    \le (K+1)(mn)^K
  $,
  and $|\Lambda(T)|\le mn+1$. For each pair $(A,\tau)$, the LP and rounding take polynomial time: $\mathrm{poly}(m,n,\log P,1/\varepsilon)$. Hence, we obtain the claimed total running time.
\end{proof}

\paragraph{Interpretation.}
\cref{thm:inner-layer-guarantee} exhibits a clean tradeoff, controlled by the budget $K$, between paying for false-negative heavy assignments in running time and paying for them in the makespan bound.
The local search step (Step 2) runs in \((mn)^{O(K)}\) time, which increases monotonically with \(K\).
In contrast, the residual false-negative term $\smash{R_{T,K}^-(\widehat{\phi},\sigma)}$ decreases monotonically with $K$.
In particular, when $K=0$, the largest missed heavy job appears directly in the bound; when $K=1$, the algorithm can recover the largest missed heavy job by search, so the bound is determined by the second-largest missed heavy job; and so on.
Hence increasing $K$ makes the algorithm slower but improves the makespan guarantee, whereas decreasing $K$ speeds up the algorithm at the cost of a weaker guarantee.
In this sense, the parameter $K$ provides a smooth interpolation between search effort and schedule quality.
For further intuition, \cref{col:inner-layer-guarantee-asymmetric,cor:inner-layer-guarantee-k-free} in \cref{apx:additional-inner-layer-analysis} discuss two special cases of \cref{thm:inner-layer-guarantee}: one where $K$ is large enough to recover all missed heavy assignments, and one where $K=0$, so no local search is performed.
The smoothness of the final makespan guarantee is discussed later in \cref{sec:outer-layer-guarantee}, where we analyze the outer-layer theorem.

\subsection{Guarantees for the outer layer}
\label{sec:outer-layer-guarantee}

\begin{theorem}[Outer-layer guarantee]
  \label{thm:outer-layer-guarantee}
  Fix $\varepsilon\in (0,1)$, a search budget $K\ge 0$, and a prediction $\predphi*$.
  Let $\sigma^\star$ be an optimal schedule, and write $\OPT \coloneqq C_{\max}(\sigma^\star)$.
  Then \cref{alg:outer-layer} returns a feasible schedule $S$ satisfying
  \begin{equation}
    \begin{aligned}
      C_{\max}(S)
      &\le
      \smash{\min\left\{
        \OPT
        + \max\{\varepsilon \OPT, R_{\OPT,K}^-(\widehat{\phi},\sigma^\star)\}
        + R_{\OPT}^+(\widehat{\phi},\sigma^\star),\,
        2\OPT
      \right\}}
      \\[9pt]
      &\le
      \smash{\min\left\{
        (1+\varepsilon)\OPT
        + R_{\OPT,K}^-(\widehat{\phi},\sigma^\star)
        + R_{\OPT}^+(\widehat{\phi},\sigma^\star),\,
        2\OPT
      \right\}.}
    \end{aligned}
  \end{equation}
  The running time of the algorithm is
  $
    \smash{O((mn)^K\mathrm{poly}(m,n,\log P,1/\varepsilon))}.
  $
\end{theorem}
\begin{proof}
  Fix an optimal schedule $\sigma^\star$ of makespan $\OPT$.
  Recall that $B$ is the baseline $2$-approximation makespan.
  Since $\OPT\!\in\! [B/2,B]$, there exists a largest critical value $T^\star\!\in\! \mathcal T$ such that $T^\star\! \le\! \OPT$.
  By construction of $\mathcal T$, there is no value of the form $p_{ij}/\varepsilon$ in the open interval $(T^\star,\OPT]$.
  Hence, for every machine--job pair $(i,j)$, we have:
  $
    p_{ij}\!\!>\!\varepsilon T^\star
    \!\!\Leftrightarrow\!\!
    p_{ij}\!>\!\varepsilon \OPT
  $.
  Thus, the heavy/short decomposition is identical at thresholds $T^\star$ and $\OPT$, and in particular we have
  $
    \smash{\widehat{\phi}_{T^\star}=\widehat{\phi}_{\OPT}},
    \smash{\mathrm{Miss}_{T^\star}(\widehat{\phi},\sigma^\star)}
    =
    \smash{\mathrm{Miss}_{\OPT}(\widehat{\phi},\sigma^\star)},
    \smash{R_{T^\star,K}^-(\widehat{\phi},\sigma^\star)}
    =
    \smash{R_{\OPT,K}^-(\widehat{\phi},\sigma^\star)}
  $
  and
  $
    \smash{R_{T^\star}^+(\widehat{\phi},\sigma^\star)}
    =
    \smash{R_{\OPT}^+(\widehat{\phi},\sigma^\star)}
  $.

  Hence we may apply Theorem~\ref{thm:inner-layer-guarantee} to the call
  $
  \smash{\textsc{Prediction-Guided-Completion}(T^\star,\varepsilon,K,\widehat{\phi})}
  $,
  taking $\sigma^\star$ as the comparison schedule.
  It follows that the returned schedule $S_{T^\star}$ satisfies
  \begin{equation}\label{eq:outer-layer-proof-t-star}
    \begin{aligned}
      C_{\max}(S_{T^\star})
      &\le
      C_{\max}(\sigma^\star)
      + \max\{\varepsilon T^\star, R_{T^\star, K}^-(\widehat{\phi},\sigma^\star)\}
      + R_{T^\star}^+(\widehat{\phi},\sigma^\star)
      \\
      &\le
      \OPT
      + \max\{\varepsilon \OPT, R_{\OPT, K}^-(\widehat{\phi},\sigma^\star)\}
      + R_{\OPT}^+(\widehat{\phi},\sigma^\star).
    \end{aligned}
  \end{equation}
  The outer algorithm returns the best schedule among the inner-layer candidates and the baseline schedule $S_0$. Since $C_{\max}(S_0)\le 2\OPT$, we obtain the claimed minimum bound.

  It remains to bound the running time. The set $\mathcal T$ contains at most $mn+2$ critical values. For each $T\in\mathcal T$, Algorithm~\ref{alg:inner-layer} runs in time
  $
  O((mn)^K\cdot \mathrm{poly}(m,n,\log P,1/\varepsilon))
  $.
  Multiplying by the polynomial number of outer iterations preserves the same asymptotic form.
  This completes the proof.
\end{proof}

\paragraph{Interpretation.}
As in \cref{{thm:inner-layer-guarantee}}, \cref{thm:outer-layer-guarantee} provides a smooth tradeoff between running time and makespan quality.
The running time grows monotonically with the budget $K$, while the makespan false-negative term $R_{\OPT,K}^-(\predphi*,\sigma^\star)$ decreases monotonically with $K$.
Thus, larger budgets $K$ let the algorithm spend more time correcting predictions and obtain stronger guarantees, whereas smaller budgets run faster at the cost of a weaker bound.

If the prediction is accurate enough that $\smash{R_{\OPT,K}^-(\widehat{\phi},\sigma^\star)=0}$ and $\smash{R_{\OPT}^+(\widehat{\phi},\sigma^\star)=0}$, the theorem gives a $(1+\varepsilon)\OPT$ schedule, matching the classical heavy-assignment search guarantee~\citep{Lenstra1990,mordechai2015optimization}.
Unlike classical schemes, which are polynomial-time only when the number of machines $m$ is fixed to be small, the exponent in our running time is governed by the search budget $K$.
This budget is the user-chosen constant parameter that controls how much the algorithm corrects the prediction, analogous to trust parameters in prior learning-augmented algorithms~\citep{NEURIPS2018_73a427ba,NEURIPS2020_e834cb11}.
Thus, the running time remains polynomial independently of whether $m$ or other problem-dependent quantities are fixed, and accurate predictions allow the algorithm to improve over the classical polynomial-time $2$-approximation baseline.

As the prediction error increases, the error terms
$\smash{R_{\OPT,K}^-(\widehat{\phi},\sigma^\star)}$ and
$\smash{R_{\OPT}^+(\widehat{\phi},\sigma^\star)}$ increase, and the makespan bound degrades smoothly with them.
Even in the worst case, the bound never exceeds the classical $2$-approximation baseline.
Thus, \cref{thm:outer-layer-guarantee} exhibits the desired learning-augmented behavior: accurate predictions yield improved guarantees, while worse predictions lead to smooth degradation without exceeding the classical worst-case guarantee.

For further intuition, we present the following corollary for the sufficiently accurate prediction.

\begin{corollary}[False-negative recovery]
  \label[corollary]{cor:outer-layer-guarantee-asymmetric}
  If
  $
    K \ge |\mathrm{Miss}_{\OPT}(\widehat{\phi},\sigma^\star)|
  $,
  \cref{alg:outer-layer} returns a feasible schedule $S$ such that
  \begin{equation}
    C_{\max}(S)
    \le
    \min\left\{
      (1 + \varepsilon)\OPT + R_{\OPT}^+(\widehat{\phi},\sigma^\star),\,
      2\OPT
    \right\}.
  \end{equation}
  The running time of the algorithm is
  $
    \smash{O((mn)^K\mathrm{poly}(m,n,\log P,1/\varepsilon))}.
  $
\end{corollary}

The proof follows immediately from \cref{thm:outer-layer-guarantee}; we provide
the details in \cref{apx:proof-of-outer-corollaries}.

\paragraph{Interpretation.}
When the budget $K$ is sufficiently large, the false-negative term $\smash{R_{\OPT,K}^-(\widehat{\phi},\sigma^\star)}$ disappears entirely from the makespan bound.
Thus, false negatives are paid for purely in running time, while false positives are paid for in the approximation bound through $\smash{R_{\OPT}^+(\widehat{\phi},\sigma^\star)}$, revealing a clear asymmetry between the two types of prediction error.
The corresponding smoothness and robustness properties are exactly those discussed for \cref{thm:outer-layer-guarantee}.

Moreover, the budget needed for this condition has a prediction-sensitive behavior reminiscent of learning-augmented guarantees.
If the prediction has no false negatives, then $K=0$ already suffices to remove the false-negative term from the makespan bound.
The budget needed to remove the false-negative term increases gradually with the prediction error and, even in the worst case, does not exceed the scale of the classical heavy-assignment search.

Note that this corollary is intended to isolate the intuition for the large-$K$ case and to highlight the asymmetry between false negatives and false positives; it does not mean that our algorithm requires $\smash{K\ge |\mathrm{Miss}_{\OPT}(\widehat{\phi},\sigma^\star)|}$.
As shown in \cref{thm:outer-layer-guarantee}, the algorithm also works for insufficient $K$, with the unrecovered false-negative term
$\smash{R_{\OPT,K}^-(\widehat{\phi},\sigma^\star)}$
appearing smoothly in the makespan bound.

As the opposite extreme, we next state the corollary for the case with no local search, \ie, $K=0$.

\begin{corollary}[No-search bound]
  \label[corollary]{cor:outer-layer-guarantee-k-free}
  When $K=0$, \cref{alg:outer-layer} returns a feasible schedule $S$ such that
  \begin{equation}
    \begin{aligned}
      C_{\max}(S)
      &\le
      \smash{\min\left\{
        \OPT + \max\{\varepsilon \OPT,\, R_{\OPT}^-(\widehat{\phi},\sigma^\star)\} + R_{\OPT}^+(\widehat{\phi},\sigma^\star),\,
        2\OPT
      \right\}}
      \\[7pt]
      &\le
      \smash{\min\left\{
        (1 + \varepsilon)\OPT + R_{\OPT}^-(\widehat{\phi},\sigma^\star) + R_{\OPT}^+(\widehat{\phi},\sigma^\star),\,
        2\OPT
      \right\}.}
    \end{aligned}
  \end{equation}
  The running time of the algorithm is
  $
    O(\mathrm{poly}(m,n,\log P,1/\varepsilon)).
  $
\end{corollary}
The proof follows immediately from \cref{thm:outer-layer-guarantee}; we provide
the details in \cref{apx:proof-of-outer-corollaries}.

\paragraph{Interpretation.}
When $K=0$, no local search is performed, so the $K$-dependent factor in the running time disappears.
Even in this case, the prediction is still used effectively, and accurate predictions can lead to a significantly improved makespan bound.
The corresponding smoothness and robustness properties are the same as those discussed for Theorem~\ref{thm:outer-layer-guarantee}.

\section{Numerical experiments}
\label{sec:experiments}

In this section, we empirically investigate how the proposed algorithm actually behaves as the prediction quality and the search budget $K$ vary.

\subsection{Experimental setup}

Following prior work on $R\|C_{\max}$ \citep{FANJULPEYRO2011301,FANJULPEYRO201055,OJMO_2021__2__A2_0}, we use synthetic instances whose processing times are sampled independently and uniformly from the interval $[1,100]$.
Following prior work on learning-augmented algorithms~\citep{NEURIPS2018_73a427ba,antoniadis2025approximation}, we generate synthetic predictions with controlled error rates to evaluate how the algorithm behaves as the prediction error varies.
Specifically, we obtain predictions by randomly corrupting a prescribed fraction of the heavy assignments in an optimal schedule into false negatives or false positives.
Throughout the experiments, we set $\varepsilon=0.01$ and vary the search budget $K$ from $0$ to $2$.
Since \citet{mordechai2015optimization} only gives a fixed-\(T\) routine, we combine it with the LST outer search as a baseline, which yields a running time of
\(O(2^{mn}\mathrm{poly}(m,n,\log P,1/\varepsilon))\).
All experiments are run with 30 different random seeds, and we report the mean and standard error.
To ensure reproducibility, the implementation of our algorithm and all experimental code are publicly available.\footnote{\url{https://github.com/kAIto47802/rcmax-with-predictions}}
Further details on the experimental setup and implementation are given in \cref{apx:detailed-experimental-setup}.

\subsection{Effect of false negatives and false positives on makespan}
\label{sec:experiment-effect-false-negatives-positives}

We use instances with $m\!=\!10$ and $n\!=\!100$, and vary either the false-negative rate or the false-positive rate while fixing the other to $0$.
The results are shown in \cref{fig:false-negative-makespan,fig:false-positive-makespan}.
Note that both classical $(1+\varepsilon)$-type baselines are not available in these experiments due to their excessive running time.

We first discuss the false-negative experiment in \cref{fig:false-negative-makespan}. Across all choices of $K$, our algorithm recovers an optimal schedule when there are no false negatives, \ie, when $\smash{R_{\OPT}^-(\widehat{\phi},\sigma^\star)/\OPT=0.0}$, and the makespan increases smoothly as the false-negative error grows.
Even in the extreme case where all heavy predictions are removed, around $\smash{R_{\OPT}^-(\widehat{\phi},\sigma^\star)/\OPT\simeq 0.38}$, the makespan remains below the LST $2$-approximation baseline (The fact that this baseline has $C_{\max}(S)/\OPT\simeq 1.26$ rather than $2$ is consistent with the fact that the $2$-approximation ratio is a worst-case guarantee).
This behavior precisely matches the learning-augmented guarantee in \cref{thm:outer-layer-guarantee}.

Comparing $K=0,1,2$, we observe that larger $K$ consistently improves the makespan.
This reflects the monotone decrease of  $R_{\OPT,K}^-(\predphi*,\sigma^\star)$ in \cref{thm:outer-layer-guarantee}.
In particular, for $K=1$ when $\smash{R_{\OPT}^-(\predphi*,\sigma^\star)/\OPT\lesssim 0.1}$, and for $K=2$ when $\smash{R_{\OPT}^-(\predphi*,\sigma^\star)/\OPT\lesssim 0.14}$, the algorithm still attains the optimal makespan despite the presence of false negatives.
This is consistent with \cref{cor:outer-layer-guarantee-asymmetric}: in these regions, the condition $K\ge |\mathrm{Miss}_{\OPT}(\predphi*,\sigma^\star)|$ is satisfied, so a small amount of search around the prediction completely eliminates the effect of false negatives from the makespan.

We next discuss the false-positive experiment in \cref{fig:false-positive-makespan}.
As in the false-negative experiment, the algorithm attains the optimal makespan when there are no false positives, and the makespan degrades smoothly as the false-positive error increases.
Even in the worst case, the makespan remains below the classical $2$-approximation baseline.
Unlike the false-negative case, all choices of $K$ yield nearly identical makespans.
This reflects \cref{thm:outer-layer-guarantee}, where the false-positive term $\smash{R_{\OPT}^+(\predphi,\sigma^\star)}$ is independent of $K$.
Thus, \cref{fig:false-negative-makespan,fig:false-positive-makespan} illustrate the theoretical asymmetry: false negatives can be recovered by search, whereas false positives remain in the makespan bound.

\cref{apx:additional-experiments-fp-fn} reports additional results for other $m,n$ settings and non-uniform instances with machine-dependent correlations in processing times, showing similar qualitative behavior consistently.

\begin{figure}[t]
  \centering
  \begin{minipage}{0.65\linewidth}
    \hspace*{-5mm}\includegraphics[height=42mm]{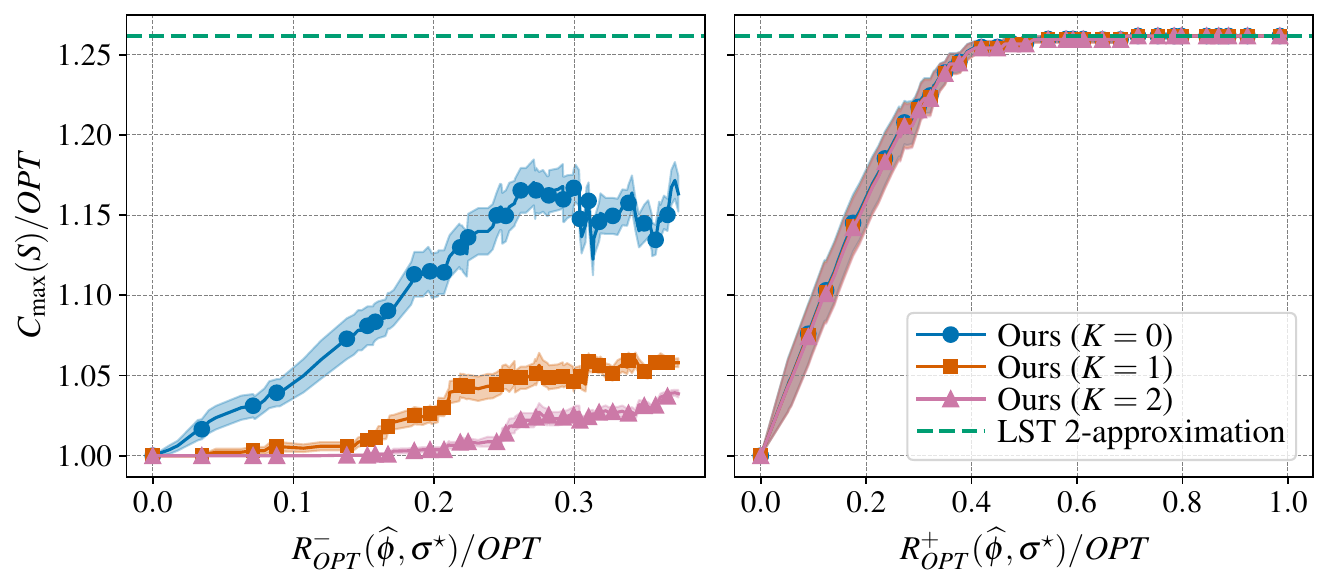}
  \end{minipage}
  \begin{minipage}{0.32\linewidth}
    \raisebox{3pt}{\includegraphics[height=40.9mm]{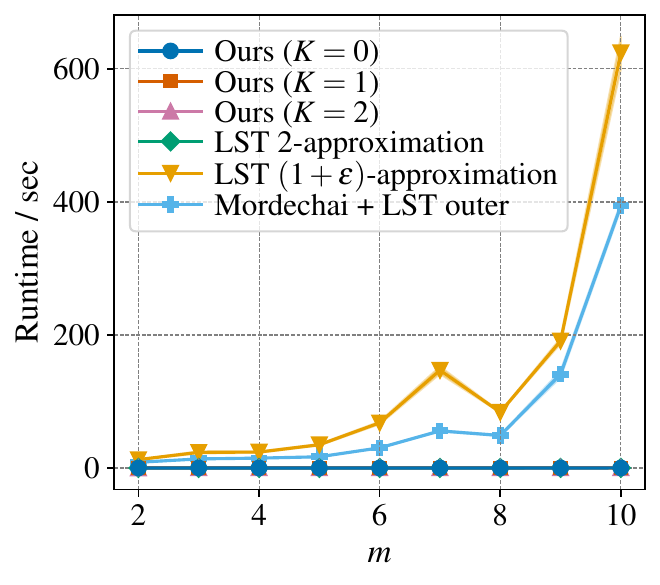}}
  \end{minipage}

  \vspace{-2mm}
  \hspace{1mm}\begin{minipage}{0.29\linewidth}
    \centering
    \subcaption{False negatives}
    \label{fig:false-negative-makespan}
  \end{minipage}
  \begin{minipage}{0.32\linewidth}
    \centering
    \subcaption{False positives}
    \label{fig:false-positive-makespan}
  \end{minipage}
  \hspace{12mm}\begin{minipage}{0.2\linewidth}
    \centering
    \subcaption{Running time}
    \label{fig:runtime-varying-m}
  \end{minipage}

  \caption{
    Numerical results.
    (a) and (b) show the change in makespan as the false-negative and false-positive rates are varied, respectively.
    (c) shows the change in running time as the number of machines \(m\) is varied; the curves for Ours with \(K=0,1,2\) and the LST \(2\)-approximation baseline overlap because they finish almost immediately.
    The lines denote the mean, and the shaded regions denote the standard error, both computed over 30 independent runs with different random seeds.
  }
  \label{fig:numerical-results}
  \vspace{3pt}
\end{figure}

\subsection{Running time comparison}

We next examine the running time as the number of machines varies.
In \cref{fig:runtime-varying-m}, we fix $n=10$ and vary $m$, with both the false-negative and false-positive rates set to $0.1$.

Our algorithm and the classical LST $2$-approximation baseline finish almost immediately even as $m$ increases.
In contrast, both classical $(1+\varepsilon)$-type baselines scale rapidly, reaching roughly $400$ and $600$ seconds even on these tiny $m,n$ settings.
This is consistent with \cref{thm:outer-layer-guarantee}: our algorithm, like the $2$-approximation baseline, has polynomial running time with constant search budget $K$, whereas the classical $(1+\varepsilon)$-type baselines are not polynomial-time unless $m$ is fixed to be small.

\cref{apx:additional-experiments-runtime-n} reports running time experiments with fixed $m$
and varying $n$, where our algorithm remains comparable to the polynomial-time
$2$-approximation baseline and much faster than the classical
$(1+\varepsilon)$-type baselines.

\section{Conclusion}

We developed a learning-augmented approximation algorithm for unrelated-machines makespan scheduling, addressing an open question of \citet{antoniadis2025approximation} on constructing smooth prediction-dependent guarantees to scheduling.
The resulting algorithm achieves a $(1+\varepsilon)\mathrm{OPT}$ guarantee when the prediction is accurate, while its makespan bound degrades smoothly with false-negative and false-positive prediction errors and never exceeds the classical $2\mathrm{OPT}$ baseline.
The search budget $K$ also provides a tunable tradeoff between running time and solution quality.
Our experiments support these theoretical properties, showing smooth degradation with prediction error, the predicted asymmetry between false negatives and false positives, and substantially better scalability than classical $(1+\varepsilon)$-type algorithms.
We also extend the prediction-localized search approach to total weighted completion time on unrelated machines, obtaining a prediction-dependent guarantee while preserving the approximation ratio of a classical baseline (\cref{apx:extension_to_completion_time}).
We hope that this prediction-localized perspective will help extend learning-augmented approximation guarantees to other optimization problems whose classical algorithms rely on expensive structural guessing.



\bibliographystyle{icml2026}
\bibliography{refs}


\appendix

\crefalias{section}{appendix}
\crefalias{subsection}{appendix}
\crefalias{subsubsection}{appendix}

\section{Extended related work}
\label{apx:extended-related-work}

\subsection{Learning-augmented online algorithms}

The learning-augmented algorithms framework was developed primarily in the online setting~\citep{10.1145/3528087}.
The central idea is to augment a worst-case algorithm with predictions, for example from a machine-learning model, and to derive guarantees that depend on the quality of these predictions.
Early representative works include the learning-augmented caching algorithm of~\citet{pmlr-v80-lykouris18a} and the algorithms of~\citet{NEURIPS2018_73a427ba} for ski rental and non-clairvoyant scheduling.
These works introduced a paradigm in which predictions can improve performance when accurate, while the algorithm retains worst-case guarantees when predictions are inaccurate.
This viewpoint is often summarized through the desiderata of consistency, robustness, and smoothness: the algorithm should perform well with accurate predictions, remain safe with poor predictions, and degrade gracefully as prediction error increases.

Subsequent work developed this framework in several directions.
For caching, \citet{doi:10.1137/1.9781611975994.112} and \citet{wei:LIPIcs.APPROX/RANDOM.2020.60} obtained sharper and simpler guarantees for learning-augmented caching, and later work studied parsimonious prediction models for the same problem~\citep{pmlr-v162-im22a}.
Beyond caching, \citet{NEURIPS2020_e834cb11} introduced a primal-dual framework for learning-augmented online algorithms and applied it to a variety of online covering problems.
This line was further generalized by \citet{NEURIPS2022_fc5a1845}, who studied learning-augmented online covering linear and semidefinite programs.
Other representative works include learning-augmented algorithms with \(\varepsilon\)-accurate predictions for problems such as caching, online set cover, online facility location, and online network design~\citep{NEURIPS2022_0ea04831}.

\subsection{Learning-augmented approximation algorithms}

A more recent line of work studies learning-augmented algorithms for offline NP-hard optimization problems.
Here the prediction is no longer about future online inputs, but may describe part of an optimal solution, a structural property of an optimal solution, or an object that a classical approximation algorithm would otherwise need to guess.

Several recent works explore this direction beyond the online setting.
\citet{NEURIPS2024_2db08b94} study Max-Cut and related max-CSPs with predictions, showing how learned information can help overcome classical approximation barriers.
\citet{pmlr-v235-bampis24a} consider parsimonious learning-augmented approximations for dense NP-hard problems.
\citet{braverman_et_al:LIPIcs.APPROX/RANDOM.2024.24} study learning-augmented Maximum Independent Set using oracle access to predicted vertex membership in a fixed maximum independent set.
\citet{pmlr-v267-bampis25a} consider NP-hard permutation problems, including scheduling and network-design problems, with predictions about the relative order of elements in an optimal permutation.
\citet{pmlr-v267-aamand25c} design improved approximation algorithms for hard graph problems using edge-based predictions, including MaxCut, Vertex Cover, Set Cover, and Maximum Independent Set.
These works highlight that the choice of prediction object is crucial: different problems require different forms of predictions.

A recent prominent work with broad applicability in offline learning-augmented approximation is \citet{antoniadis2025approximation}.
They give black-box algorithms for a broad class of selection problems with linear objectives, where a feasible subset of items is selected to minimize or maximize total weight.
Their framework yields results for problems such as Vertex Cover, Steiner Tree, Min-Weight Perfect Matching, Knapsack, and Clique.
However, scheduling problems, such as \(R\|C_{\max}\) lie outside this black-box setting, as they are assignment problems with a makespan objective rather than total-weight selection problems, and were explicitly identified as a natural open direction for future work.
Our work takes a step in this direction by developing a learning-augmented algorithmic framework tailored to the heavy-assignment structure of unrelated-machine scheduling.

\subsection{Classical approximation algorithms for \(R\|C_{\max}\)}
The unrelated-machines makespan problem \(R\|C_{\max}\) is a central problem in scheduling.
The classical algorithm of~\citet{Lenstra1990} gives a polynomial-time \(2\)-approximation and shows that no polynomial-time approximation with ratio below \(3/2\) is possible unless \(P=NP\).
The same work also gives a polynomial approximation scheme for the case of a fixed number of machines, which yields a \((1+\varepsilon)\OPT\) guarantee.
The scheme uses a heavy/short decomposition: for a target makespan~\(T\), assignments with processing time larger than \(\varepsilon T\) are treated as large assignments, and there can be only a bounded number of such assignments on each machine in a schedule of makespan at most \(T\).
Thus, if the large assignments are guessed correctly, the remaining jobs are short and can be completed by solving an LP and applying the LST rounding procedure.

\citet{mordechai2015optimization} later studied parameterized approximation schemes for scheduling, including unrelated-machines makespan.
In this approach, the parameter is the number of large machine-job pairs $k(T)=|\{(i,j)\mid p_{ij}>\varepsilon T\}|$.
The algorithm treats the variables corresponding to these large pairs integrally and rounds the remaining fractional solution.
This yields a \((1+\varepsilon)\OPT\) guarantee with running time exponential in \(k(T)\) and polynomial in the input size.

Our algorithm can be viewed as a prediction-localized refinement of this paradigm: instead of searching over the full heavy-assignment structure, it searches only for the heavy assignments that are missing from the prediction, yielding prediction-sensitive tradeoffs in both running time and approximation quality.

\subsection{Learning and scheduling}
There is also some work using machine learning to improve the practical performance of scheduling algorithms and solvers.
\citet{BOUSKA2023990} use a neural network within a decomposition-based algorithm for the single-machine total tardiness problem, where the model guides the decomposition by estimating the best split.
Relatedly, \citet{hijazi2024needimprovingcolumnenhancing} use transformer-style models to enhance column generation for parallel-machine scheduling by predicting improving columns.
These works show that ML predictions can be useful in scheduling practice, especially when they guide expensive search or decomposition steps.
However, their focus is primarily empirical or solver-acceleration oriented.
Also, neither of these works studies \(R\|C_{\max}\) itself.
In contrast, our work gives approximation-theoretic guarantees: the prediction determines a local neighborhood of the classical heavy-assignment search, and the resulting makespan and running time bounds depend explicitly on prediction errors.

\section{Choice of Prediction Model}
\label{apx:prediction_choice}

Predicting which jobs should be assigned to which machines is a particularly natural prediction model for this problem, and it directly targets the main combinatorial difficulty in the classical approximation scheme.

For comparison, we discuss below several alternative prediction objects and their limitations.

\begin{itemize}
  \item Full job-to-machine assignment: This is already covered as a special case of our model (as we wrote in Remark 4.1). Allowing partial predictions is more flexible, since the predictor can abstain on uncertain jobs and leave them to the algorithm rather than fixing them incorrectly.
  \item Set of heavy jobs: This is insufficient for unrelated machines because heaviness is a property of a machine-job pair $(i,j)$, rather than of a job alone.
  \item Optimal makespan or the load of every machine in an optimal schedule: These predictions provide only aggregate information and do not identify the combinatorial heavy-assignment structure. Even if the exact value $T=\mathrm{OPT}$ were known, classical $(1+\varepsilon)$-approximation schemes would still need to search for the heavy assignments.
  \item Fractional LP solution for the full instance: A fractional solution does not directly provide the required integral heavy-assignment structure. In particular, rounding fractional mass on heavy edges can incur a large additive loss, so it cannot replace the heavy-assignment prediction without additional structural assumptions.
  \item Processing-time prediction: In our offline setting, every processing time $p_{ij}$ is already given as part of the input. Predicting processing times therefore provides no additional information; this would instead correspond to a different problem in which processing times are unknown or stochastic.
\end{itemize}

\section{Proofs deferred from the main text}

\subsection{Full proof of \cref{thm:inner-layer-guarantee}}\label{apx:full-proof-of-inner-layer}

In \cref{sec:prediction-error-measures}, we presented the proof sketch of \cref{thm:inner-layer-guarantee}. Here we provide the full proof.

\begin{proof}
  By the definition of $R_{T,K}^-(\widehat{\phi},\sigma)$, there exists a subset
  $U^\star\subseteq \mathrm{Miss}_{T}(\widehat{\phi},\sigma)$ with $|U^\star|\le K$ such that every missed heavy job
  outside \(U^\star\) has processing time at most
  \(R_{T,K}^-(\widehat{\phi},\sigma)\), \ie,
  $
    \max_{j\in \mathrm{Miss}_{T}(\widehat{\phi},\sigma)\setminus U^\star} p_{\sigma(j),j}
    =
    R_{T,K}^-(\widehat{\phi},\sigma)
  $.
  Since $U^\star\subseteq \mathrm{Miss}_{T}(\widehat{\phi},\sigma)$, every job in $U^\star$ is missed by
  $\widehat{\phi}_T$ and satisfies $p_{\sigma(j),j}>\varepsilon T$.

  Define an augmentation
  \begin{equation}
    A^\star(j)\coloneqq
    \begin{cases}
    \sigma(j), & \text{if } j\in U^\star,\\[-3pt]
    \bot, & \text{otherwise.}
    \end{cases}
  \end{equation}
  Then we have
  $
    |\operatorname{dom}(A^\star)|=|U^\star|\le K
  $ by construction.
  Thus, $A^\star$ satisfies \cref{eq:augmentation-constraint} and is considered by the algorithm.

  Let
  $
    \psi^\star \coloneqq \widehat{\phi}_T \cup A^\star
  $
  and define
  $
    \tau^\star
    \coloneqq
    \max\{\varepsilon T,\,R_{T,K}^-(\widehat{\phi},\sigma)\}
  $.
  By construction, $\tau^\star$ is one of the threshold candidate set $\Lambda(T)$ considered by the algorithm:
  if $R_{T,K}^-(\widehat{\phi},\sigma)=0$, then $\tau^\star=\varepsilon T$;
  otherwise, $R_{T,K}^-(\widehat{\phi},\sigma)$ is the processing time
  $p_{\sigma(j),j}$ of some missed heavy job $j$.

  Now we claim that the residual LP with threshold $\tau^\star$ for $\psi^\star$ is feasible with
  load bound
  $
    L\coloneqq C_{\max}(\sigma)+R_T^+(\widehat{\phi},\sigma)
  $.
  Consider the fractional assignment that sends every residual job
  $j\notin \operatorname{dom}(\psi^\star)$ to machine $\sigma(j)$.
  We verify feasibility of this assignment for the residual LP with threshold $\tau^\star$.

  Let $j\notin \operatorname{dom}(\psi^\star)$.
  If $p_{\sigma(j),j}\le \varepsilon T$, then clearly
  $p_{\sigma(j),j}\le \tau^\star$.
  Otherwise, $p_{\sigma(j),j}>\varepsilon T$. Since $j\notin \operatorname{dom}(\psi^\star)$,
  we have $\widehat{\phi}_T(j)=\bot$ and $j\notin U^\star$.
  Hence $j\in \mathrm{Miss}_{T}(\predphi*,\sigma)\setminus U^\star$, and therefore we have
  $
    p_{\sigma(j),j}
    \le
    R_{T,K}^-(\widehat{\phi},\sigma)
    \le
    \tau^\star
  $.
  Thus, in all cases, the edge $(\sigma(j),j)$ is allowed in the
  residual LP with threshold $\tau^\star$.

  We now verify the machine-capacity constraints.
  Fix a machine $i$.
  Under the assignment that follows $\psi^\star$ on fixed jobs and $\sigma$ on residual jobs, the total load on machine $i$ is
  \begin{equation}\label{eq:total-load-on-machine-i-under-sigma}
    \begin{aligned}
      \ell_i^{\mathrm{fix}}(\psi^\star)
      +
      \hspace{-8pt}\sum_{\substack{j:\sigma(j)=i,\\ \psi^\star(j)=\bot}}\!\! p_{ij}
      &=
      \hspace{-10pt}\sum_{\substack{j:\sigma(j)=i\\ \psi^\star(j)\in\{\bot,i\}}}\hspace{-10pt} p_{ij}
      +
      \hspace{-8pt}\sum_{\substack{j:\psi^\star(j)=i\\ \sigma(j)\neq i}}\hspace{-4pt} p_{ij}
      =
      \hspace{-10pt}\sum_{\substack{j:\sigma(j)=i\\ \psi^\star(j)\in\{\bot,i\}}}\hspace{-10pt} p_{ij}
      +
      \hspace{-10pt}\sum_{\substack{j:\widehat{\phi}_T(j)=i\\ \widehat{\phi}_T(j)\neq \sigma(j)}}\hspace{-8pt} p_{ij}
      \\
      &\le
      \hspace{-5pt}\sum_{j:\sigma(j)=i}\! p_{ij}
      +
      \hspace{-10pt}\sum_{\substack{j:\widehat{\phi}_T(j)=i\\ \widehat{\phi}_T(j)\neq \sigma(j)}}\hspace{-8pt} p_{ij}
      \le
      C_{\max}(\sigma) + R_T^+(\widehat{\phi},\sigma)
      =: L.
    \end{aligned}
  \end{equation}
  Here, the first equality uses
  \begin{equation}
    \ell_i^{\mathrm{fix}}(\psi^\star)
    =
    \!\!\sum_{j:\psi^\star(j)=i} p_{ij}
    =
    \!\!\sum_{\substack{j:\psi^\star(j)=i\\ \sigma(j)=i}} p_{ij}
    +
    \!\!\sum_{\substack{j:\psi^\star(j)=i\\ \sigma(j)\neq i}} p_{ij},
  \end{equation}
  and the second equality holds because every job in $A^\star$ is assigned according to $\sigma$ by construction, so the only fixed jobs that may increase the load of machine $i$ relative to $\sigma$ are those coming from $\widehat{\phi}_T$.

  Thus, from \cref{eq:total-load-on-machine-i-under-sigma}, if we define
  \begin{equation}
    x_{ij}=
    \begin{cases}
    1, & \text{if } i=\sigma(j),\\[-3pt]
    0, & \text{otherwise},
    \end{cases}
    \; j\notin \operatorname{dom}(\psi^\star),
  \end{equation}
  we have
  \begin{equation}
    \sum_{j\notin \operatorname{dom}(\psi^\star)}\hspace{-6pt} p_{ij}x_{ij}
    =
    \!\!\! \sum_{\substack{j:\sigma(j)=i,\\ \psi^\star(j)=\bot}}\!\! p_{ij}
    \le
    L-\ell_i^{\mathrm{fix}}(\psi^\star).
  \end{equation}
  Hence the residual jobs, when assigned according to $\sigma$, form a feasible
  solution to the residual LP with load bound $L$.
  In particular, the residual LP is feasible and its optimum value satisfies
  $
    L_{\psi^\star}^{\tau^\star}\le L
  $.

  By the classical rounding argument (Theorem 2.1 in \citet{Shmoys1993}), the algorithm
  can round this residual solution into an integral schedule of makespan at most
  $
    L_{\psi^\star}^{\tau^\star}+\tau^\star
  $.
  Therefore, for the candidate pair $(A^\star,\tau^\star)$, the algorithm obtains
  a schedule of makespan at most
  \begin{equation}
    C_{\max}(\sigma)
    +
    R_T^+(\widehat{\phi},\sigma)
    +
    \max\{\varepsilon T,\,R_{T,K}^-(\widehat{\phi},\sigma)\}.
  \end{equation}
  Since Algorithm~\ref{alg:inner-layer} returns the best schedule among all candidates,
  its output satisfies the same bound.

  It remains to analyze the running time.
  The computation of $\widehat{\phi}_T$ takes polynomial time.
  The number of augmentations considered by the algorithm is at most
  \begin{equation}
    \sum_{s=0}^{K} \binom{n}{s}m^s
    \le
    (K+1)(mn)^K
    \le
    (n+1)(mn)^K.
  \end{equation}
  For each augmentation $A$, the algorithm iterates over the threshold set $\Lambda(T)$,
  whose size is at most $mn+1$.
  For each pair $(A,\tau)$, it forms the combined assignment, solves the residual
  LP, and, if feasible, applies LST rounding. Each of these steps takes time
  polynomial in $m$, $n$, $\log P$, and $1/\varepsilon$.
  Hence the total running time is
  $
    O((mn)^{K}\cdot \mathrm{poly}(m,n,\log P,1/\varepsilon))
  $.
\end{proof}

\subsection{Proofs of \cref{cor:outer-layer-guarantee-asymmetric,cor:outer-layer-guarantee-k-free}}
\label{apx:proof-of-outer-corollaries}

\begin{proof}[Proof of \cref{cor:outer-layer-guarantee-asymmetric}]
  Since $K \ge |\mathrm{Miss}_{\OPT}(\widehat{\phi},\sigma^\star)|$, all missed $\varepsilon \OPT$-heavy assignments can be recovered by the augmentation step. Hence
  $
    R_{\OPT,K}^-(\widehat{\phi},\sigma^\star)=0
  $.
  Applying \cref{thm:outer-layer-guarantee}, we obtain the claimed bound immediately.
  The running-time bound is inherited directly from \cref{thm:outer-layer-guarantee}.
\end{proof}

\begin{proof}[Proof of \cref{cor:outer-layer-guarantee-k-free}]
  Set $K=0$ in \cref{thm:outer-layer-guarantee}.
  Since we have $
  R_{\OPT,0}^-(\widehat{\phi},\sigma^\star)
  =
  R_{\OPT}^-(\widehat{\phi},\sigma^\star)
  $ by definition,
  \cref{thm:outer-layer-guarantee} gives the claimed bound immediately.
  The running-time bound also follows immediately by substituting $K=0$ into the running time of \cref{thm:outer-layer-guarantee}.
\end{proof}

\section{Additional discussions}

\subsection{Additional discussions for the inner layer}
\label{apx:additional-inner-layer-analysis}

The following corollary is obtained as the special case of \cref{thm:inner-layer-guarantee} in which the search budget $K$ is sufficiently large.

\begin{corollary}
  \label[corollary]{col:inner-layer-guarantee-asymmetric}
  Suppose there exists a schedule $\sigma$ such that
  $
    K \ge |\mathrm{Miss}_T(\widehat{\phi},\sigma)|
  $.
  Then, \cref{alg:inner-layer} returns a feasible schedule $S$ satisfying
  \begin{equation}
  C_{\max}(S)
  \le
  C_{\max}(\sigma) + \varepsilon T + R_T^+(\widehat{\phi},\sigma).
  \end{equation}
  The running time of the algorithm is
  $
    O((mn)^K\mathrm{poly}(m,n,\log P,1/\varepsilon)).
  $
\end{corollary}

\begin{proof}
  Since $K \ge |\mathrm{Miss}_T(\widehat{\phi},\sigma)|$, the algorithm can recover all missed heavy assignments. Hence, by the definition, we have
  $
  R_{T,K}^-(\widehat{\phi},\sigma)=0
  $.
  Applying \cref{thm:inner-layer-guarantee}, we obtain the claimed bound immediately.
  The running-time bound is inherited directly from Theorem~\ref{thm:inner-layer-guarantee}.
\end{proof}

\paragraph{Interpretation.}
\cref{col:inner-layer-guarantee-asymmetric} reveals a genuine asymmetry between the two types of prediction error.
False-negative heavy assignments affect the running time through the augmentation budget $K$, since they must be recovered by additional local search.
On the other hand, false-positive heavy assignments do not enlarge the search space, but instead consume machine capacity and therefore appear directly in the makespan bound through $R_T^+(\widehat{\phi},\sigma)$.

The required running time to satisfy the condition degrades smoothly with the amount of false-negative heavy predictions, since the required augmentation budget $K$ grows with the number of missed-heavy jobs $|\mathrm{Miss}_T(\widehat{\phi},\sigma)|$.
At the same time, even when the number of false negatives is large, the search remains in the same asymptotic regime as the classical algorithm based on exhaustive search over heavy assignments~\citep{Lenstra1990,mordechai2015optimization}.

On the quality side, if the false-positive heavy assignments are negligible, then the algorithm achieves a $(1 + \varepsilon)OPT$ makespan guarantee, matching the classical completion-based guarantee with exhaustive search~\citep{Lenstra1990,mordechai2015optimization}. As the false-positive load increases, the makespan guarantee deteriorates smoothly through the additive overload term $R_T^+(\widehat{\phi},\sigma)$.

As the opposite extreme, when no search is performed, \ie, $K=0$, we obtain the following corollary.

\begin{corollary}
  \label[corollary]{cor:inner-layer-guarantee-k-free}
  When $K=0$, for any schedule $\sigma$, \cref{alg:inner-layer} returns a feasible schedule $S$ satisfying
  \begin{equation}
  C_{\max}(S)
  \le
  C_{\max}(\sigma) + \max\{\varepsilon T,\, R_T^-(\widehat{\phi},\sigma)\} + R_T^+(\widehat{\phi},\sigma).
  \end{equation}
  In particular, if $\sigma$ has makespan at most $T$, then the returned schedule $S$ satisfies
  \begin{equation}
    \begin{aligned}
      C_{\max}(S)
      &\le
      T + \max\{\varepsilon T,\, R_T^-(\widehat{\phi},\sigma)\} + R_T^+(\widehat{\phi},\sigma)\\
      &\le
      (1 + \varepsilon)T + R_T^-(\widehat{\phi},\sigma) + R_T^+(\widehat{\phi},\sigma).
    \end{aligned}
  \end{equation}
  The running time of the algorithm is
  $
    O(\mathrm{poly}(m,n,\log P,1/\varepsilon)).
  $
\end{corollary}

\begin{proof}
  Apply \cref{thm:inner-layer-guarantee} with $K=0$.
  Since
  $
  R_{T,0}^-(\widehat{\phi},\sigma)
  =
  R_T^-(\widehat{\phi},\sigma)
  $
  by definition, we obtain the claimed makespan bound immediately.
  The running-time bound also follows immediately from \cref{thm:inner-layer-guarantee}, which gives
  $
  O((mn)^0\cdot \mathrm{poly}(m,n,\log P,1/\varepsilon))
  =
  O(\mathrm{poly}(m,n,\log P,1/\varepsilon))
  $.
\end{proof}

\paragraph{Interpretation.}
Corollary~\ref{cor:inner-layer-guarantee-k-free} eliminates the search requirement on $K$.
Instead of recovering missed heavy assignments combinatorially, the algorithm absorbs them into the approximation bound through $R_T^-(\widehat{\phi},\sigma)$, making the $(mn)^{O(K)}$ search disappear entirely.
Here, the upper bound deteriorates smoothly with both false-negative term $R_T^-(\widehat{\phi},\sigma)$ and false-positive term $R_T^+(\widehat{\phi},\sigma)$.

\section{A variant with editable predictions}
\label{apx:editable-variant}

The formulation we presented in \cref{sec:algorithm} trusts every predicted assignment that survives the $T$-heavy projection.
We now describe a more flexible variant in which the algorithm is allowed to delete, insert, or edit a bounded number of heavy predicted assignments.
This variant uses the same residual LP and rounding
procedure as before; the only change is the local neighborhood of heavy partial assignments searched by the inner layer.

\subsection{Editable heavy-assignment neighborhood}
\label{apx:editable-neighborhood}

We first define the neighborhood of heavy partial assignments considered by the inner layer.
Fix a threshold $T>0$. Recall that $\widehat{\phi}_T$ denotes the $T$-heavy
projection of the prediction. For a partial assignment
$\psi:[n]\to [m]\cup\{\bot\}$, define its edit distance from the projected
prediction at scale $T$ by
\begin{equation}
  d_T(\psi,\widehat{\phi})
  :=
  \left|
    \left\{
      j\in[n]
      \;\middle|\;
      \psi(j)\neq \widehat{\phi}_T(j)
    \right\}
  \right|.
\end{equation}
Thus, changing $\bot$ to a machine corresponds to an insertion, changing a predicted machine to $\bot$ corresponds to a deletion, and changing a predicted machine to another one corresponds to an edit.

For an edit budget $K\ge 0$, define the editable $K$-neighborhood of the
prediction by
\begin{equation}
  \mathcal{E}_{T,K}(\widehat{\phi})
  :=
  \left\{
    \psi\colon[n]\to[m]\cup\{\bot\}
    \;\middle|\;
    d_T(\psi,\widehat{\phi})\le K,\;
    \psi(j)\neq\bot
    \Rightarrow
    \varepsilon T < p_{\psi(j),j}
    \right\}.
\end{equation}
The editable inner-layer algorithm enumerates all
$\psi\in\mathcal{E}_{T,K}(\widehat{\phi})$ in Step 2, in place of all augmentations defined in \cref{eq:augmentation-constraint}.
The rest of the algorithm is unchanged.
For each such $\psi$, it fixes the
jobs in $\operatorname{dom}(\psi)$, solves the same residual LP over thresholds
$\tau\in\Lambda(T)$, applies Shmoys--Tardos rounding whenever the LP is
feasible, and returns the best schedule among all candidates.

\subsection{Edit-aware prediction error measures}

We next define prediction-error measures for the editable variant. The key
point is that deletions and edits may trade false positives against false
negatives: deleting an incorrectly fixed heavy assignment can reduce
false-positive overload, but if the job is truly heavy in the comparison
schedule, it may become a missed heavy assignment. Therefore, the two components
must be defined with respect to the same edited prediction.

Accordingly, we use the same false-negative and false-positive error measures as in
\cref{eq:fn-fp-error-measures}, but evaluate them on the edited partial
assignment $\psi$ instead of the projected prediction $\widehat{\phi}_T$.
Concretely, fix an arbitrary minimizer
\begin{equation}
\psi^\mathrm{edit}_{T,K}(\widehat{\phi},\sigma)
\in
\operatorname*{arg\,min}_{\psi\in\mathcal{E}_{T,K}(\widehat{\phi})}
\left[
\max\{\varepsilon T,R^-_T(\psi,\sigma)\}
+
R^+_T(\psi,\sigma)
\right],
\end{equation}
where ties are broken deterministically.
We then define the edit-aware
false-negative and false-positive components by
\begin{equation}
R^{-,\mathrm{edit}}_{T,K}(\widehat{\phi},\sigma)
:=
R^-_T
\left(
\psi^\mathrm{edit}_{T,K}(\widehat{\phi},\sigma),
\sigma
\right),
\qquad
R^{+,\mathrm{edit}}_{T,K}(\widehat{\phi},\sigma)
:=
R^+_T
\left(
\psi^\mathrm{edit}_{T,K}(\widehat{\phi},\sigma),
\sigma
\right).
\end{equation}

Equivalently, it is simpler to define the edit-aware combined prediction error directly as
\begin{equation}
\Delta^\mathrm{edit}_{T,K}(\widehat{\phi},\sigma)
:=
\min_{\psi\in\mathcal{E}_{T,K}(\widehat{\phi})}
\left[
\max\{\varepsilon T,R^-_T(\psi,\sigma)\}
+
R^+_T(\psi,\sigma)
\right].
\end{equation}

By definition, we have
\begin{equation}
\Delta^\mathrm{edit}_{T,K}(\widehat{\phi},\sigma)
=
\max\left\{
\varepsilon T,
R^{-,\mathrm{edit}}_{T,K}(\widehat{\phi},\sigma)
\right\}
+
R^{+,\mathrm{edit}}_{T,K}(\widehat{\phi},\sigma).
\end{equation}
The quantity $\Delta^\mathrm{edit}_{T,K}$ is monotone nonincreasing in the
budget $K$, although
$R^{-,\mathrm{edit}}_{T,K}$ and $R^{+,\mathrm{edit}}_{T,K}$ need not be
individually monotone because an edit may trade one type of error for the other.


\subsection{Inner-layer guarantee for the editable variant}

\begin{theorem}[Inner-layer guarantee with editable predictions]
\label{thm:inner-edit}
Fix $\varepsilon\in(0,1)$, a target value $T>0$, an edit budget $K\ge 0$, and a
prediction $\widehat{\phi}$.
Consider the editable inner-layer algorithm that
enumerates all $\psi\in\mathcal{E}_{T,K}(\widehat{\phi})$ in place of all augmentations defined in \cref{eq:augmentation-constraint}.

Then, for any schedule $\sigma$, the algorithm returns a feasible schedule $S$
satisfying
\begin{equation}
C_{\max}(S)
\le
C_{\max}(\sigma)
+
\max\left\{
\varepsilon T,
R^{-,\mathrm{edit}}_{T,K}(\widehat{\phi},\sigma)
\right\}
+
R^{+,\mathrm{edit}}_{T,K}(\widehat{\phi},\sigma).
\end{equation}
Equivalently,
\begin{equation}
C_{\max}(S)
\le
C_{\max}(\sigma)
+
\Delta^\mathrm{edit}_{T,K}(\widehat{\phi},\sigma).
\end{equation}
The running time of the algorithm is
$
O\!\left(
(mn)^K \operatorname{poly}(m,n,\log P,1/\varepsilon)
\right).
$
\end{theorem}

\begin{proof}
Let
$
\psi^\star
:=
\psi^\mathrm{edit}_{T,K}(\widehat{\phi},\sigma)
$
be an edited partial assignment attaining the definition of the edit-aware
error measures. Since
$\psi^\star\in\mathcal{E}_{T,K}(\widehat{\phi})$, it is one of the candidates
enumerated by the editable inner-layer algorithm.

Define
$
\tau^\star
:=
\max\left\{
\varepsilon T,
R^{-,\mathrm{edit}}_{T,K}(\widehat{\phi},\sigma)
\right\}
$.
We first note that $\tau^\star\in\Lambda(T)$. If
$R^{-,\mathrm{edit}}_{T,K}(\widehat{\phi},\sigma)=0$, then
$\tau^\star=\varepsilon T$. Otherwise,
$R^{-,\mathrm{edit}}_{T,K}(\widehat{\phi},\sigma)$ is equal to
$p_{\sigma(j),j}$ for some missed heavy job $j$, and hence appears in the discrete threshold set $\Lambda(T)$.

Consider the residual LP for the fixed assignment $\psi^\star$ and threshold
$\tau^\star$. Assign every residual job $j\notin\operatorname{dom}(\psi^\star)$
to the machine $\sigma(j)$. We verify that each such edge is allowed. If
$p_{\sigma(j),j}\le \varepsilon T$, then
$p_{\sigma(j),j}\le \tau^\star$. Otherwise
$p_{\sigma(j),j}>\varepsilon T$. Since $j$ is residual under $\psi^\star$, we
have $\psi^\star(j)=\bot$, and therefore
$j\in\operatorname{Miss}_T(\psi^\star,\sigma)$. Hence
$
p_{\sigma(j),j}
\le
R^-_T(\psi^\star,\sigma)
=
R^{-,\mathrm{edit}}_{T,K}(\widehat{\phi},\sigma)
\le
\tau^\star
$.
Thus the assignment of residual jobs according to $\sigma$ uses only edges allowed by the residual LP.

We next bound the load. Fix a machine $i$. Under the assignment that follows
$\psi^\star$ on fixed jobs and follows $\sigma$ on residual jobs, the load on
machine $i$ is
\begin{equation}
  \ell^{\mathrm{fix}}_i(\psi^\star)
  +
  \sum_{\substack{
  j:\sigma(j)=i\\
  \psi^\star(j)=\bot
  }}
  p_{ij}
  =
  \!\!\!\sum_{\substack{
  j:\sigma(j)=i\\
  \psi^\star(j)\in\{\bot,i\}
  }}
  \!\!\!p_{ij}
  +
  \!\!\!\sum_{\substack{
  j:\psi^\star(j)=i\\
  \psi^\star(j)\neq \sigma(j)
  }}
  \!\!\!p_{ij}.
  \end{equation}

  \vspace{-12pt}
  Therefore,
  \begin{equation}
  \ell^{\mathrm{fix}}_i(\psi^\star)
  +
  \!\!\!\sum_{\substack{
  j:\sigma(j)=i\\
  \psi^\star(j)=\bot
  }}
  \!\!\!p_{ij}
  \le
  \!\!\sum_{j:\sigma(j)=i}
  \!\!p_{ij}
  +
  R^+_T(\psi^\star,\sigma)
  \le
  C_{\max}(\sigma)
  +
  R^{+,\mathrm{edit}}_{T,K}(\widehat{\phi},\sigma).
\end{equation}
Hence the residual LP for $(\psi^\star,\tau^\star)$ is feasible with load bound
at most
$
C_{\max}(\sigma)
+
R^{+,\mathrm{edit}}_{T,K}(\widehat{\phi},\sigma)
$.

By the Shmoys--Tardos rounding guarantee, rounding the residual LP solution
increases the load of every machine by at most $\tau^\star$. Thus, for the
candidate $\psi^\star$, the algorithm obtains a schedule of makespan at most

\vspace{-15pt}
\begin{equation}
C_{\max}(\sigma)
+
R^{+,\mathrm{edit}}_{T,K}(\widehat{\phi},\sigma)
+
\max\left\{
\varepsilon T,
R^{-,\mathrm{edit}}_{T,K}(\widehat{\phi},\sigma)
\right\}.
\end{equation}
Since the algorithm returns the best candidate, the same bound holds for its
output.

It remains to bound the running time. The number of edited partial assignments
in $\mathcal{E}_{T,K}(\widehat{\phi})$ is at most
$
\sum_{s=0}^{K}
\binom{n}{s}(m+1)^s
=
O((mn)^K).
$
For each candidate $\psi$, the algorithm considers at most $mn+1$ thresholds in
$\Lambda(T)$, and for each pair $(\psi,\tau)$ it solves one LP and performs one
rounding step, both in polynomial time. This gives the claimed running time.
\end{proof}

\paragraph{Interpretation.}
\cref{thm:inner-edit} shows that allowing edits changes the role of the
budget $K$ from a false-negative recovery budget into a general correction budget for the predicted heavy-assignment structure.
In the original augmentation-only variant, the algorithm can only insert missed heavy assignments; hence false positives remain fixed and appear directly in the makespan bound.
In contrast, the editable variant can spend part of its budget to delete an incorrect prediction or to edit it to another machine.
Thus, both false negatives and false positives can be repaired by local search.

The guarantee should be interpreted through the combined error
$\smash{\Delta^{\mathrm{edit}}_{T,K}(\widehat{\phi},\sigma)}$.
This quantity is monotone nonincreasing in $K$, because a larger budget only enlarges the edit neighborhood $\smash{\mathcal{E}_{T,K}(\widehat{\phi})}$.
However, the two components
$R^{-,\mathrm{edit}}_{T,K}(\widehat{\phi},\sigma)$ and
$R^{+,\mathrm{edit}}_{T,K}(\widehat{\phi},\sigma)$ need not be individually monotone.
For example, deleting a false-positive heavy assignment may reduce
false-positive overload, but if that job is heavy in the comparison schedule,
it may become a missed heavy assignment. This is why the theorem defines the
edit-aware error by minimizing the sum
\begin{equation}
  \max\{\varepsilon T,R^-_T(\psi,\sigma)\}+R^+_T(\psi,\sigma)
\end{equation}
over a common edited partial assignment $\psi$.

Consequently, $K$ provides a smooth tradeoff between running time and the
quality of the edited prediction. Larger values of $K$ allow the algorithm to
search a larger neighborhood of the prediction and can only improve the
combined prediction-dependent bound, while increasing the running time by the
factor $(mn)^K$. When $K=0$, the editable variant reduces to trusting the
projected prediction $\predphi_T$ without any local correction; as $K$
increases, the algorithm gradually interpolates toward the best completion
obtainable after correcting up to $K$ heavy predicted assignments.

\subsection{Outer-layer guarantee for the editable variant}

We next show that the outer-layer guarantee for the augmentation-only variant
(\cref{thm:outer-layer-guarantee}) also extends to the editable variant with the same proof structure, by simply replacing the augmentation-based error measures with their edit-aware counterparts.

\begin{theorem}[Outer-layer guarantee with editable predictions]
\label{thm:outer-edit}
Fix $\varepsilon\in(0,1)$, an edit budget $K\ge 0$, and a prediction $\widehat{\phi}$. Let $\sigma^\star$ be an optimal schedule, and write $\OPT:=C_{\max}(\sigma^\star)$.
The editable outer-layer algorithm returns a feasible schedule $S$ satisfying
\begin{equation}
  C_{\max}(S)
  \le
  \min\left\{
  \OPT
  +
  \max\left\{
  \varepsilon\OPT,
  R^{-,\mathrm{edit}}_{\OPT,K}
  (\widehat{\phi},\sigma^\star)
  \right\}
  +
  R^{+,\mathrm{edit}}_{\OPT,K}
  (\widehat{\phi},\sigma^\star),
  \;
  2\OPT
  \right\}.
\end{equation}
Equivalently,
\begin{equation}
  \smash{
    C_{\max}(S)
    \le
    \min\left\{
    \OPT
    +
    \Delta^\mathrm{edit}_{\OPT,K}
    (\widehat{\phi},\sigma^\star),
    \;
    2\OPT
    \right\}.
  }
\end{equation}
In particular,
\begin{equation}
  C_{\max}(S)
  \le
  \min\left\{
  (1+\varepsilon)\OPT
  +
  R^{-,\mathrm{edit}}_{\OPT,K}
  (\widehat{\phi},\sigma^\star)
  +
  R^{+,\mathrm{edit}}_{\OPT,K}
  (\widehat{\phi},\sigma^\star),
  \;
  2\OPT
  \right\}.
\end{equation}
The running time of the algorithm is
$
  O\!\left(
  (mn)^K \operatorname{poly}(m,n,\log P,1/\varepsilon)
  \right)
$.
\end{theorem}

\vspace{-8pt}

\begin{proof}
Let $S_0$ be the baseline schedule computed by the classical
$2$-approximation, and let $B=C_{\max}(S_0)$. Then
$\OPT\in[B/2,B]$. Let $T^\star$ be the largest critical threshold in
the outer search satisfying $T^\star\le \OPT$.

By construction of the critical threshold set, there is no value of the form
$p_{ij}/\varepsilon$ in the open interval $(T^\star,\OPT]$. Therefore,
for every machine-job pair $(i,j)$, we have the equivalence
$
p_{ij}>\varepsilon T^\star
\Leftrightarrow
p_{ij}>\varepsilon\OPT
$.
Consequently, the heavy/short decomposition is identical at thresholds
$T^\star$ and $\OPT$. In particular, we have
$
\widehat{\phi}_{T^\star}
=
\widehat{\phi}_\OPT$, and
$
\mathcal{E}_{T^\star,K}(\widehat{\phi})
=
\mathcal{E}_{\OPT,K}(\widehat{\phi})
$.

Let
$
\psi^\circ
:=
\psi^\mathrm{edit}_{\OPT,K}
(\widehat{\phi},\sigma^\star)
$
be an edited partial assignment attaining the edit-aware error at scale
$\OPT$. Since
$\mathcal{E}_{T^\star,K}(\widehat{\phi})
=
\mathcal{E}_{\OPT,K}(\widehat{\phi})$, the inner call at threshold
$T^\star$ enumerates $\psi^\circ$.

Moreover, because the heavy/short decomposition is identical at
$T^\star$ and $\OPT$, the primitive missed-heavy set and the primitive
false-positive overload of $\psi^\circ$ with respect to $\sigma^\star$ are the
same at the two thresholds:
\begin{equation}
  \smash{R^-_{T^\star}(\psi^\circ,\sigma^\star)
  =
  R^-_\OPT(\psi^\circ,\sigma^\star),
  \qquad
  R^+_{T^\star}(\psi^\circ,\sigma^\star)
  =
  R^+_\OPT(\psi^\circ,\sigma^\star).}
\end{equation}
Applying the same inner-layer argument to the candidate $\psi^\circ$ at
threshold $T^\star$ gives a schedule $S_{T^\star}$ satisfying
\begin{equation}
  \smash{
  C_{\max}(S_{T^\star})
  \le
  \OPT
  +
  \max\left\{
  \varepsilon T^\star,
  R^-_\OPT(\psi^\circ,\sigma^\star)
  \right\}
  +
  R^+_\OPT(\psi^\circ,\sigma^\star).}
\end{equation}

Since $T^\star\le\OPT$, we have
$
\max\left\{
\varepsilon T^\star,
R^-_\OPT(\psi^\circ,\sigma^\star)
\right\}
\le
\max\left\{
\varepsilon\OPT,
R^-_\OPT(\psi^\circ,\sigma^\star)
\right\}
$.
By the definition of $\psi^\circ$, we have
$
R^-_\OPT(\psi^\circ,\sigma^\star)
=
R^{-,\mathrm{edit}}_{\OPT,K}
(\widehat{\phi},\sigma^\star)$
and
$
R^+_\OPT(\psi^\circ,\sigma^\star)
=
R^{+,\mathrm{edit}}_{\OPT,K}
(\widehat{\phi},\sigma^\star)
$.
Therefore, it follows that
\begin{equation}
  C_{\max}(S_{T^\star})
  \le
  \OPT
  +
  \max\left\{
  \varepsilon\OPT,
  R^{-,\mathrm{edit}}_{\OPT,K}
  (\widehat{\phi},\sigma^\star)
  \right\}
  +
  R^{+,\mathrm{edit}}_{\OPT,K}
  (\widehat{\phi},\sigma^\star).
\end{equation}
The outer algorithm returns the best schedule among all inner-layer candidates
and the baseline schedule $S_0$. Since
$C_{\max}(S_0)\le 2\OPT$, the claimed bound follows.

Finally, the outer search considers at most $mn+2$ critical thresholds. Combining this with the running-time bound from Theorem~\ref{thm:inner-edit} gives the claimed running time.
\end{proof}

\vspace{-4pt}

\paragraph{Interpretation.}
Compared with the augmentation-only guarantee in \cref{thm:outer-layer-guarantee}, the editable guarantee is
strictly more flexible. False-negative heavy assignments can be inserted, false-positive heavy assignments can be deleted, and incorrect heavy assignments can be edited to another machine.
Therefore, the relevant error is not merely the number or size of missed heavy assignments, nor merely the overload caused by trusted false positives, but rather the smallest residual combination of these two effects after at most $K$ edits.
This residual combination is exactly captured by $\Delta^{\mathrm{edit}}_{\mathrm{OPT},K}(\widehat{\phi},\sigma^\star)$.

As the prediction becomes more accurate, the edit-aware error $\Delta^{\mathrm{edit}}_{\mathrm{OPT},K}(\widehat{\phi},\sigma^\star)$ decreases and the guarantee approaches $(1+\varepsilon)\mathrm{OPT}$.
As the prediction becomes worse, the bound degrades smoothly through the same error term, but it never exceeds $2\mathrm{OPT}$ regardless of the prediction quality.
Thus the editable variant preserves the same consistency--smoothness--robustness behavior as the main algorithm, while allowing the search budget to correct both types of prediction error.

\paragraph{Relation to the augmentation-only guarantee (\cref{thm:outer-layer-guarantee}).}
The editable guarantee generalizes the augmentation-only guarantee. Indeed, the
augmentation-only search considered earlier corresponds to the restricted
subfamily of $\mathcal{E}_{T,K}(\widehat{\phi})$ in which the only allowed
changes are insertions of missed heavy assignments, i.e., changes of the form
$\bot\to i$. For the candidate that inserts the best $K$ missed heavy
assignments according to the comparison schedule $\sigma$, the primitive
false-negative threshold is exactly $R^-_{T,K}(\widehat{\phi},\sigma)$ and the
primitive false-positive overload remains $R^+_T(\widehat{\phi},\sigma)$.
Since the editable error minimizes over a larger family of corrected partial
assignments, we have
\begin{equation}
  \smash{\Delta^{\mathrm{edit}}_{T,K}(\widehat{\phi},\sigma)
  \le
  \max\{\varepsilon T,R^-_{T,K}(\widehat{\phi},\sigma)\}
  +
  R^+_T(\widehat{\phi},\sigma).}
\end{equation}
Equivalently,
\begin{equation}
  \smash{\max\left\{
  \varepsilon T,
  R^{-,\mathrm{edit}}_{T,K}(\widehat{\phi},\sigma)
  \right\}
  +
  R^{+,\mathrm{edit}}_{T,K}(\widehat{\phi},\sigma)
  \le
  \max\{\varepsilon T,R^-_{T,K}(\widehat{\phi},\sigma)\}
  +
  R^+_T(\widehat{\phi},\sigma).}
\end{equation}
Thus, allowing deletions and edits cannot worsen the prediction-dependent bound at
the same asymptotic running time.

\subsection{Corollary for consistency under small edit distance}

For further intuition, we state the following consistency corollary for the case where the prediction is sufficiently accurate.

\begin{corollary}[Consistency under small edit distance]
\label[corollary]{cor:edit-consistency}
Let $\sigma^\star$ be an optimal schedule. Define its heavy projection at scale
$\OPT$ by
\begin{equation}
\sigma^\star_\OPT(j)
:=
\begin{cases}
\sigma^\star(j),
&
\text{if } p_{\sigma^\star(j),j}>\varepsilon\OPT,\\
\bot,
&
\text{otherwise}.
\end{cases}
\end{equation}
Define the heavy edit distance between the prediction and the optimum by
\begin{equation}
d^\mathrm{edit}_\OPT
(\widehat{\phi},\sigma^\star)
:=
\left|
\left\{
j\in[n]:
\widehat{\phi}_\OPT(j)
\neq
\sigma^\star_\OPT(j)
\right\}
\right|.
\end{equation}
If
$
K
\ge
d^\mathrm{edit}_\OPT
(\widehat{\phi},\sigma^\star)
$,
then the editable outer-layer algorithm returns a feasible schedule $S$
satisfying
\begin{equation}
C_{\max}(S)
\le
(1+\varepsilon)\OPT.
\end{equation}
\end{corollary}

\vspace{-8pt}

\begin{proof}
The condition
$K\ge d^\mathrm{edit}_\OPT(\widehat{\phi},\sigma^\star)$ implies
that
$
\sigma^\star_\OPT
\in
\mathcal{E}_{\OPT,K}(\widehat{\phi}).
$
For this edited partial assignment, there are no missed
$\varepsilon\OPT$-heavy jobs and no false-positive fixed assignments
with respect to $\sigma^\star$. Hence we have
$
R^-_\OPT
(\sigma^\star_\OPT,\sigma^\star)=0,
R^+_\OPT
(\sigma^\star_\OPT,\sigma^\star)=0
$, and therefore
$
\Delta^\mathrm{edit}_{\OPT,K}
(\widehat{\phi},\sigma^\star)
\le
\varepsilon\OPT
$.
Applying Theorem~\ref{thm:outer-edit} gives
$
C_{\max}(S)
\le
\min\{(1+\varepsilon)\OPT,2\OPT\}
=
(1+\varepsilon)\OPT
$,
where the equality uses $\varepsilon\in(0,1)$.
\end{proof}

\paragraph{Interpretation.}
\cref{cor:edit-consistency} identifies the strongest consistency regime for the
editable variant. If the predicted heavy partial assignment is within $K$ edits
of the heavy projection of an optimal schedule, the algorithm incurs only an additive
$\varepsilon\OPT$ loss.

This differs from the augmentation-only consistency statement in \cref{cor:outer-layer-guarantee-asymmetric}.
In the augmentation-only variant, a sufficiently large budget can remove all false negatives, but any false positives that survive the $T$-heavy projection remain fixed and still contribute to the makespan bound.
On the other hand, in the editable variant,
false positives can also be corrected, either by deletion or by editing the predicted machine. Hence, when the entire heavy partial assignment is within the edit budget, both residual error terms vanish:
\begin{equation}
  R^{-,\mathrm{edit}}_{\mathrm{OPT},K}
  (\widehat{\phi},\sigma^\star)
  =
  R^{+,\mathrm{edit}}_{\mathrm{OPT},K}
  (\widehat{\phi},\sigma^\star)
  =
  0.
\end{equation}
The algorithm therefore achieves the clean
$(1+\varepsilon)\mathrm{OPT}$ guarantee.

As in \cref{cor:outer-layer-guarantee-asymmetric}, the required budget here is also prediction-sensitive.
If the prediction already matches the optimal heavy projection, then $K=0$ is sufficient. If the prediction has a small number of incorrect, missing, or extra heavy assignments, the required budget increases by exactly the heavy edit distance. Even in the worst case, the required budget does not exceed the one needed in the classical $(1+\varepsilon)$-type baselines.

\section{Extension to total weighted completion time}
\label{apx:extension_to_completion_time}

The idea of using predictions to localize the assignment search also extends to the unrelated-machines total weighted completion-time problem $R\|\sum_j w_j C_j$, another NP-hard scheduling problem~\citep{Skutella2001}.
Following the $K$-edit approach in \cref{apx:editable-variant}, we search over assignments that differ from a full assignment prediction on at most $K$ jobs.
For each candidate assignment, the optimal processing order on each machine can be computed efficiently using Smith's rule~\citep{repec:wly:navlog:v:3:y:1956:i:1-2:p:59-66}.
Combining this local search with a classical approximation algorithm yields a guarantee that improves with prediction accuracy while preserving the baseline worst-case approximation ratio.

\subsection{Problem setting}

We consider the unrelated-machines total weighted completion-time problem $R||\sum_j w_j C_j$.

Let $M=[m]$ be a set of machines and $J=[n]$ a set of jobs.
Job $j$ has processing time $p_{ij}>0$ on machine $i$ and weight $w_j>0$.
A feasible schedule assigns every job to one machine and specifies a nonpreemptive processing order on each machine.
Let $C_j$ denote the completion time of job $j$.
The objective is to minimize
\begin{equation}
  \operatorname{WCT}(S)\coloneqq\sum_{j\in J}w_jC_j.
\end{equation}

\subsection{Prediction model}

For simplicity, we consider a full assignment prediction $\widehat{\phi}:J\to M$, where $\widehat{\phi}(j)=i$ predicts that job $j$ should be assigned to machine $i$.

\subsection{Cost of a fixed assignment}

For any fixed assignment $\psi:J\to M$, the optimal processing order on each machine is given by Smith's rule~\citep{repec:wly:navlog:v:3:y:1956:i:1-2:p:59-66}: jobs assigned to machine $i$ are processed in nonincreasing order of $w_j/p_{ij}$.

For a set $S\subseteq J$, let $F_i(S)$ denote the weighted completion-time cost of processing exactly the jobs in $S$ on machine $i$ in Smith order. For distinct jobs $j,k$, define
\begin{equation}
  q_i(j,k)\coloneqq\min\{p_{ij}w_k,p_{ik}w_j\}.
\end{equation}

Then, expanding the completion times into individual processing costs and
pairwise waiting costs gives
\begin{equation}
  F_i(S)=\sum_{j\in S} w_jp_{ij}+\sum_{\{j,k\}\subseteq S} q_i(j,k).
\end{equation}

For an assignment $\psi:J\to M$, let $J_i(\psi):=\{j\in J\mid\psi(j)=i\}$ denote the set of jobs assigned to machine $i$.
The minimum weighted completion-time cost among schedules consistent with $\psi$ is
\begin{equation}
  \operatorname{WCT}(\psi)=\sum_{i\in M}F_i(J_i(\psi)).
\end{equation}

Thus, once an assignment is fixed, its optimal sequencing and objective value can be computed in polynomial time.

\subsection{Algorithm}

Fix an edit budget $K\ge 0$.
For two assignments $\psi,\widehat{\phi}:J\to M$, define their
assignment edit distance by
\begin{equation}
  d_H(\psi,\widehat{\phi})
  \coloneqq
  \left|
  \{j\in J\mid \psi(j)\neq\widehat{\phi}(j)\}
  \right|.
\end{equation}
The $K$-edit neighborhood of the prediction is
\begin{equation}
  \mathcal{E}_K(\widehat{\phi})
  \coloneqq
  \left\{
  \psi:J\to M
  \;\middle|\;
  d_H(\psi,\widehat{\phi})\le K
  \right\}.
\end{equation}

The prediction-localized algorithm enumerates every assignment
$\psi\in\mathcal{E}_K(\widehat{\phi})$.
For each candidate assignment, it schedules the jobs on each machine
according to Smith's rule and computes $\operatorname{WCT}(\psi)$.
In parallel, it runs a polynomial-time $\alpha$-approximation
algorithm for $R\|\sum_j w_j C_j$.
The algorithm returns the minimum-cost schedule among all enumerated
candidates and the schedule returned by the $\alpha$-approximation
algorithm.

\subsection{Prediction-error measure}

Fix a reference assignment $\sigma:J\to M$.
For each machine $i$, define $C_i(\psi,\sigma):=J_i(\psi)\cap J_i(\sigma)$.
Thus, $C_i(\psi,\sigma)$ contains the jobs assigned to machine $i$
by both assignments.

Define the completion-time error of $\psi$ relative to $\sigma$ by
\begin{equation}
\Delta^{\mathrm{CT}}(\psi,\sigma):=\sum_{i\in M}\left[F_i(J_i(\psi))-F_i(C_i(\psi,\sigma))\right].
\end{equation}

The residual error after at most $K$ edits is
\begin{equation}
  \Delta^{\mathrm{CT}}_K(\widehat{\phi},\sigma)
  \coloneqq
  \min_{\psi\in\mathcal{E}_K(\widehat{\phi})}
  \Delta^{\mathrm{CT}}(\psi,\sigma).
\end{equation}
Since
$\mathcal{E}_K(\widehat{\phi})
\subseteq
\mathcal{E}_{K+1}(\widehat{\phi})$,
this quantity is monotone nonincreasing with respect to $K$.

\subsection{Analysis}
\begin{theorem}[Prediction-localized guarantee for weighted completion time]
  For any reference assignment $\sigma$, the prediction-localized algorithm returns a schedule $S$ satisfying
  \begin{equation}
    \operatorname{WCT}(S)
    \le
    \min\left\{
    \operatorname{WCT}(\sigma)
    +
    \Delta^{\mathrm{CT}}_K(\widehat{\phi},\sigma),
    \;
    \alpha\operatorname{OPT}
    \right\},
  \end{equation}
  where $\alpha$ is the approximation ratio of the baseline algorithm.

  In particular, for an optimal assignment $\sigma^\star$, the returned schedule satisfies
  \begin{equation}
    \operatorname{WCT}(S)
    \le
    \min\left\{
    \operatorname{OPT}
    +
    \Delta^{\mathrm{CT}}_K(\widehat{\phi},\sigma^\star),
    \;
    \alpha\operatorname{OPT}
    \right\}.
  \end{equation}
  The running time is
  \begin{equation}
    O\!\left(
    (mn)^K\operatorname{poly}(m,n,\log P,\log W)
    \right).
  \end{equation}
  where $P:=\max_{i\in M,j\in J}p_{ij}$ and $W:=\max_{j\in J}w_j$.
\end{theorem}

\begin{proof}
  Let
  $\psi^\star\in\arg\min_{\psi\in\mathcal{E}_K(\widehat{\phi})}\Delta^{\mathrm{CT}}(\psi,\sigma).$
  Since the algorithm enumerates all assignments in
  $\mathcal{E}_K(\widehat{\phi})$,
  the assignment $\psi^\star$ is one of its candidates.

  For every machine $i$, we have
  \begin{equation}
    C_i(\psi^\star,\sigma)=J_i(\psi^\star)\cap J_i(\sigma)\subseteq J_i(\sigma).
  \end{equation}
  The pairwise expression for $F_i$ consists only of nonnegative terms.
  It follows that $F_i$ is monotone with respect to set inclusion, and hence
  \begin{equation}
    F_i\bigl(C_i(\psi^\star,\sigma)\bigr)\le F_i\bigl(J_i(\sigma)\bigr).
  \end{equation}
  Therefore,
  \begin{equation}
    \begin{aligned}
      \operatorname{WCT}(\psi^\star)
      &=
      \sum_{i\in M}F_i\bigl(J_i(\psi^\star)\bigr)\\
      &=
      \sum_{i\in M}
      F_i\bigl(C_i(\psi^\star,\sigma)\bigr)
      +
      \Delta^{\mathrm{CT}}(\psi^\star,\sigma)\\
      &\le
      \sum_{i\in M}F_i\bigl(J_i(\sigma)\bigr)
      +
      \Delta^{\mathrm{CT}}_K(\widehat{\phi},\sigma)\\
      &=
      \operatorname{WCT}(\sigma)
      +
      \Delta^{\mathrm{CT}}_K(\widehat{\phi},\sigma).
    \end{aligned}
  \end{equation}
  Since the algorithm returns the best schedule among the enumerated
  candidate and the baseline schedule, the same bound holds for the returned schedule $S$.

  The number of assignments in the edit neighborhood is at most
  $\sum_{s=0}^{K}\binom{n}{s}(m-1)^s = O((mn)^K)$.
  For each candidate assignment, the Smith ordering and its objective
  value can be computed in time polynomial in
  $m,n,\log P,\log W$. The baseline algorithm also runs in polynomial
  time. Hence, the total running time is
  $O\!\left(
  (mn)^K\operatorname{poly}(m,n,\log P,\log W)
  \right)$.
\end{proof}

\begin{corollary}[Consistency under small edit distance]
If $d_H(\widehat{\phi},\sigma^\star)\le K$ for some optimal assignment $\sigma^\star$, then the algorithm returns an optimal schedule.
\end{corollary}

\begin{proof}
If $d_H(\widehat{\phi},\sigma^\star)\le K$, then $\sigma^\star\in\mathcal{E}_K(\widehat{\phi})$.
The algorithm therefore considers the optimal assignment
$\sigma^\star$, and applying Smith's rule on each machine produces a schedule of cost $\operatorname{OPT}$.
\end{proof}

\paragraph{Interpretation}

The theorem gives a prediction-dependent tradeoff between running time and schedule quality.
Increasing the edit budget $K$ enlarges the neighborhood $\mathcal{E}_K(\widehat{\phi})$, and therefore $\Delta^{\mathrm{CT}}_K(\widehat{\phi},\sigma)$ is monotone nonincreasing in $K$, while the running time grows as $(mn)^K$.

If the prediction is within $K$ edits of an optimal assignment, the optimal assignment is included in the search neighborhood and the algorithm returns an optimal schedule. As the residual prediction error increases, the guarantee degrades through $\Delta^{\mathrm{CT}}_K$, while the $\alpha$-approximation baseline ensures that the resulting cost never exceeds the classical worst-case guarantee.

\section{Detailed experimental setup}
\label{apx:detailed-experimental-setup}

\paragraph{Prediction generation.}
We generate synthetic predictions from an optimal schedule \(\sigma^\star\).
We first keep only the heavy assignments of \(\sigma^\star\), namely assignments with processing time greater than \(\varepsilon \OPT\), and set all other jobs to \(\bot\).
To introduce false negatives, we randomly select a \(\rho^-\) fraction of these heavy jobs and change their predicted assignment to \(\bot\).
To introduce false positives, we restrict attention to heavy jobs that have at least one alternative heavy machine \(i\neq \sigma^\star(j)\) with \(p_{ij}>\varepsilon \OPT\); we then randomly select a \(\rho^+\) fraction of such jobs and reassign each of them to an alternative heavy machine.
When varying the error level, we sweep \(\rho^-\) or \(\rho^+\) over
\(\{0,0.001,0.002,\ldots,1\}\), and keep the predictions that yield distinct values of
\(R_{\OPT}^-(\widehat{\phi},\sigma^\star)\) or
\(R_{\OPT}^+(\widehat{\phi},\sigma^\star)\).

\paragraph{Implementation details.}
All experiments are implemented in Python 3.13.11.
We compute optimal schedules using the CP-SAT solver in OR-Tools~\citep{cpsatlp}.
The LPs used in our algorithm and in the classical baseline are solved using the HiGHS solver~\citep{Huangfu_Parallelizing_the_dual_2018} through SciPy~\citep{2020SciPy-NMeth}.
We use OR-Tools version 9.15.6755 and SciPy version 1.17.1.

\paragraph{Machine specifications.}
All experiments are run on a single machine with an Intel\textsuperscript{\textregistered} Core\textsuperscript{TM} Ultra 7 155H CPU (22 logical CPUs) and 16\,GB RAM, running Arch Linux.

\section{Additional experimental results}
\label{apx:additional-experiments}

\subsection{Additional results on false negatives and false positives}
\label{apx:additional-experiments-fp-fn}

In \cref{sec:experiment-effect-false-negatives-positives}, we reported experiments for the setting with $m=10$ and $n=100$, where the processing times $p_{ij}$ are generated independently and uniformly at random.
In this appendix, we provide additional experiments under different choices of $m$ and $n$, and under a non-uniform processing-time model.

\begin{figure}[t!]
  \centering

  \hspace*{-8mm}\includegraphics[width=0.8\textwidth]{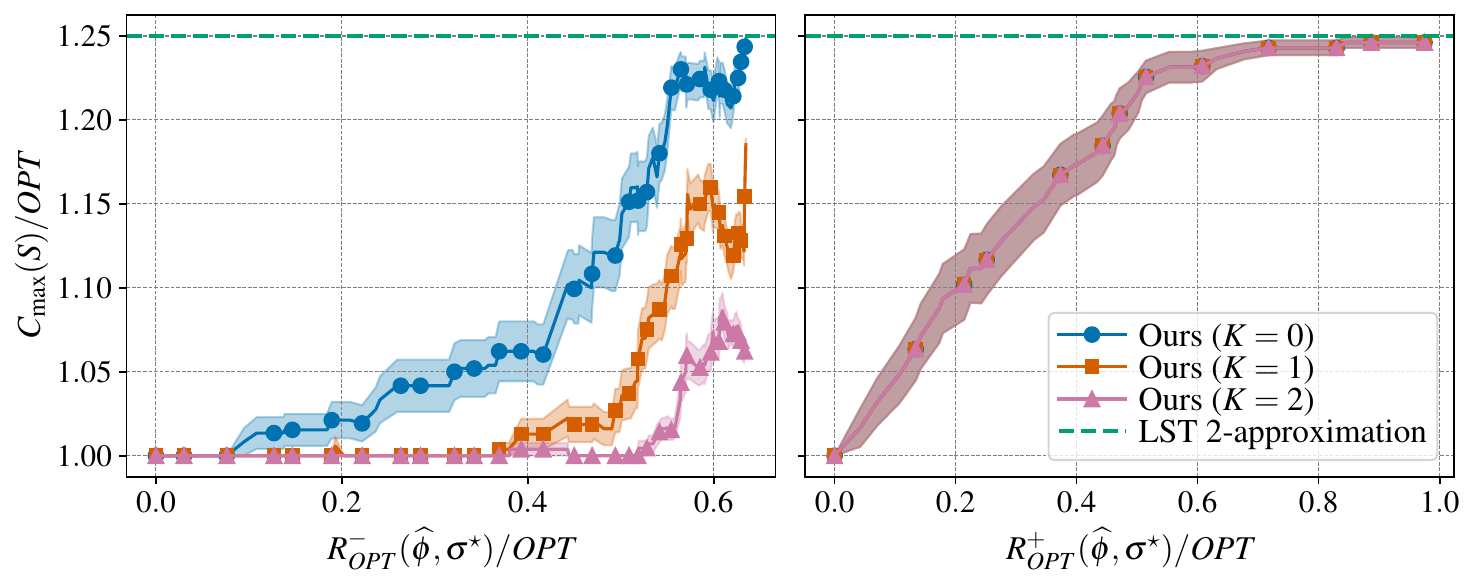}

  \vspace{-2mm}
  \hspace*{2mm}\begin{minipage}{0.36\linewidth}
    \centering
    \subcaption{
      \centering
      $m=10$, $n=50$, i.i.d. uniform;\\
      false negatives
    }
    \label{fig:apx-fn-10x50}
  \end{minipage}
  \begin{minipage}{0.36\linewidth}
    \centering
    \subcaption{
      \centering
      $m=10$, $n=50$, i.i.d. uniform;\\
      false positives
    }
    \label{fig:apx-fp-10x50}
  \end{minipage}

  \vspace{2.5mm}
  \hspace*{-8mm}\includegraphics[width=0.8\textwidth]{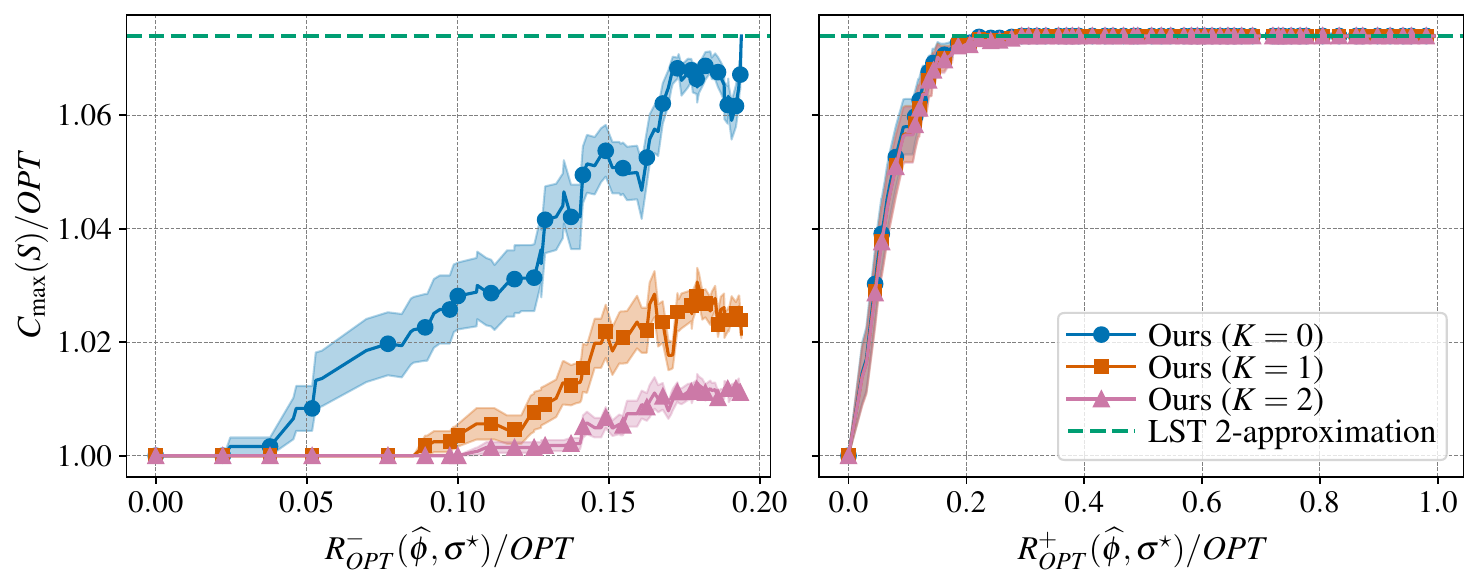}

  \vspace{-2mm}
  \hspace*{2mm}\begin{minipage}{0.36\linewidth}
    \centering
    \subcaption{
      \centering
      $m=5$, $n=100$, i.i.d. uniform;\\
      false negatives
    }
    \label{fig:apx-fn-5x100}
  \end{minipage}
  \begin{minipage}{0.36\linewidth}
    \centering
    \subcaption{
      \centering
      $m=5$, $n=100$, i.i.d. uniform;\\
      false positives
    }
    \label{fig:apx-fp-5x100}
  \end{minipage}

  \vspace{2.5mm}
  \hspace*{-8mm}\includegraphics[width=0.8\textwidth]{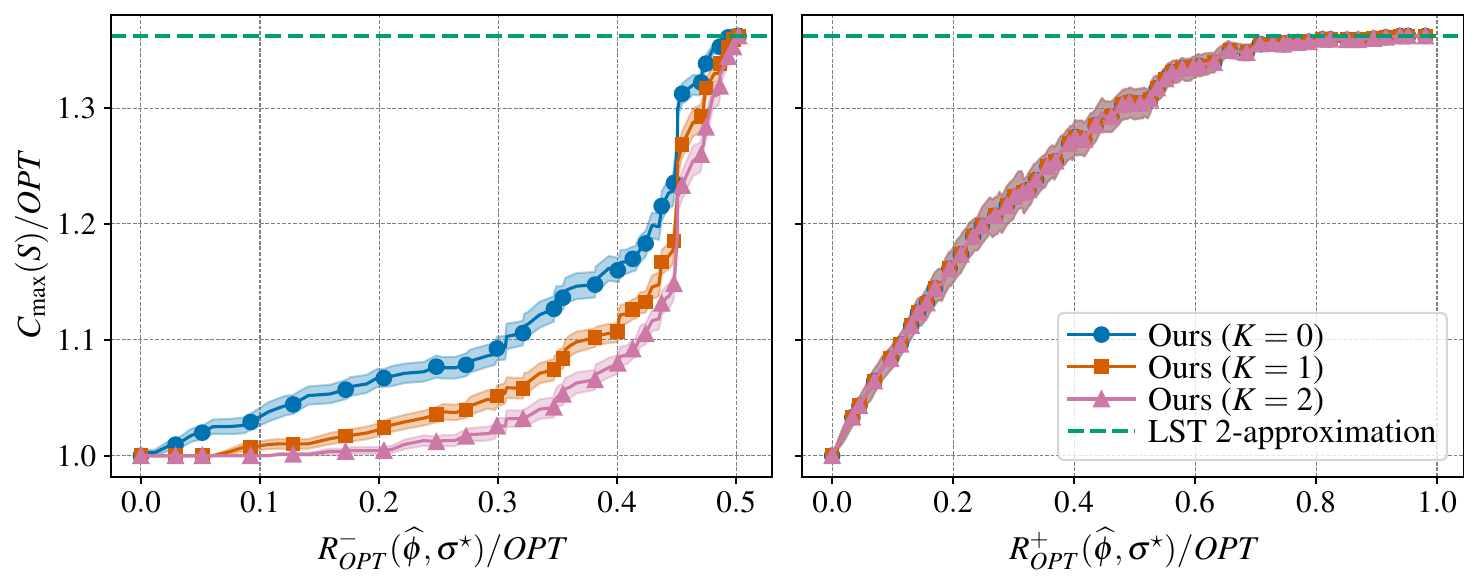}

  \vspace{-2mm}
  \begin{minipage}{0.39\linewidth}
    \centering
    \subcaption{
      \centering
      $m=10$, $n=100$, machine-correlated;\\
      false negatives
    }
    \label{fig:apx-fn-corr}
  \end{minipage}
  \begin{minipage}{0.39\linewidth}
    \centering
    \subcaption{
      \centering
      $m=10$, $n=100$, machine-correlated;\\
      false positives
    }
    \label{fig:apx-fp-corr}
  \end{minipage}

  \caption{
    Additional numerical results under alternative experimental settings.
    Panels (a) and (b) show the results for $m=10$ and $n=50$ with independently
    and uniformly generated processing times.
    Panels (c) and (d) show the results for $m=5$ and $n=100$ under the same
    generation model.
    Panels (e) and (f) show the results for $m=10$ and $n=100$ with
    machine-correlated processing times.
    In each row, the left panel varies the false-negative rate and the right
    panel varies the false-positive rate.
    The lines denote the mean, and the shaded regions denote the standard error,
    both computed over 30 independent runs with different random seeds.
  }
  \label{fig:additional-numerical-results-smoothness}
\end{figure}

\subsubsection{Experimental settings}

We consider the following three settings:
\begin{itemize}
  \item $m=10$ and $n=50$, with processing times $p_{ij}$ generated
  independently and uniformly at random as in
  \cref{sec:experiment-effect-false-negatives-positives}.

  \item $m=5$ and $n=100$, with processing times $p_{ij}$ generated
  independently and uniformly at random as in
  \cref{sec:experiment-effect-false-negatives-positives}.

  \item $m=10$ and $n=100$, with processing times generated from a
  machine-correlated model.
\end{itemize}
For the machine-correlated setting, following prior work on $R\|C_{\max}$
\citep{FANJULPEYRO2011301,FANJULPEYRO201055,OJMO_2021__2__A2_0}, we sample
$
  \alpha_i \sim \operatorname{Uniform}\{1,\ldots,100\}
$
and
$
  \eta_{ij} \sim \operatorname{Uniform}\{1,\ldots,20\}
$
independently for all $i\in[m]$ and $j\in[n]$, and set
\begin{equation}
  p_{ij}\coloneqq\alpha_i+\eta_{ij}.
\end{equation}
This model induces correlations among processing times associated with the
same machine through the shared machine-specific component $\alpha_i$.

\subsubsection{Results and discussion}
\label{apx:additional-experiments-results}

The results are shown in \cref{fig:additional-numerical-results-smoothness}.
As in \cref{sec:experiment-effect-false-negatives-positives}, both classical $(1+\varepsilon)$-type baselines were computationally infeasible in these experiments due to their excessive running time.

Across all instance classes, the results show the same qualitative behavior.
For both false-negative and false-positive errors, our algorithm attains the optimal makespan when the prediction error is zero, and its makespan degrades smoothly as the error increases. The resulting schedules remain no worse than the classical $2$-approximation baseline.
In the false-negative setting, increasing the search budget $K$ further improves the makespan, reflecting the fact that additional local search can recover more missed heavy assignments. These trends are consistent with the results in \cref{sec:experiment-effect-false-negatives-positives} and provide additional empirical support for the effectiveness of our approach.

\subsection{Additional running-time results with varying \(n\)}
\label{apx:additional-experiments-runtime-n}

\begin{wrapfigure}{r}{0.5\textwidth}
  \vspace{-\intextsep}
  \centering
  \includegraphics[width=\linewidth]{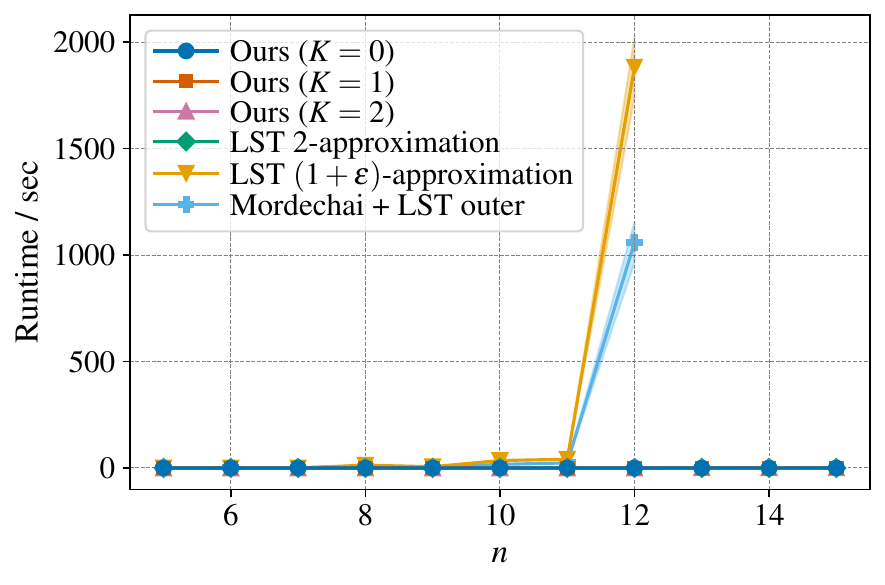}

  \vspace{-2mm}
  \caption{
    Additional running-time results with fixed \(m=5\) and varying \(n\).
    The classical \((1+\varepsilon)\)-type baselines are shown only up to
    \(n=12\) because of their large running time.
    The lines denote the mean, and the shaded regions denote the standard error,
    both computed over five independent runs with different random seeds.
  }
  \label{fig:additional-runtime-varying-n}
\end{wrapfigure}

We also report additional running-time results in which the number of machines
is fixed and the number of jobs is varied. Specifically, we fix \(m=5\) and vary
$
  n\in\{5,6,\allowbreak7,8,9,10,11,12,13,14,15\}.
$
The results are shown in \cref{fig:additional-runtime-varying-n}. Since the
classical \((1+\varepsilon)\)-type baselines become computationally expensive,
we report their results only up to \(n=12\).

The results show that our algorithm remains substantially faster than the
classical \((1+\varepsilon)\)-type baselines even when \(m\) is fixed and \(n\)
varies. This reflects the fact that the classical \((1+\varepsilon)\)-type
running times can still contain large factors such as \(n^{m/\varepsilon}\) or
\(2^{mn}\), even for constant \(m\). In contrast, our algorithm runs almost as
fast as the polynomial-time \(2\)-approximation baseline throughout this range,
which is consistent with our theoretical guarantee for constant search budget
\(K\).

\section{Limitations and broader impact}
\label{sec:limitations}

\paragraph{Limitations.}
Our results are developed for unrelated-machines makespan scheduling with predictions of heavy assignments.
While this prediction model is well aligned with the heavy/short decomposition underlying classical approximation schemes, other scheduling settings may require different prediction objects.
Our experiments use synthetic instances and controlled prediction errors, and evaluating the method with learned predictions on application-specific scheduling data remains an important direction for future work.

\paragraph{Broader Impact.}
The potential positive impact of this work is to provide a principled way to combine increasingly accurate machine-learning predictions with classical approximation guarantees for scheduling problems, a setting not covered by the framework of \citet{antoniadis2025approximation}.
As prediction quality continues to improve, such approaches enable scheduling algorithms to substantially improve practical performance beyond classical worst-case baselines, while retaining robustness when predictions are inaccurate.

We do not identify direct negative societal impacts specific to this theoretical
work, since it does not introduce a deployed system, collect data, or make
decisions about individuals. As with any scheduling system, however, deployment
in application-specific settings should validate the quality of predictions and
the consequences of scheduling decisions.



\end{document}